\documentclass[letterpaper,11pt]{article}
\usepackage{latexsym}
\usepackage{amssymb}
\usepackage[margin=1in]{geometry}
\usepackage{amsmath,amsfonts,amsthm}
\usepackage{graphicx}
\usepackage{multirow}
\usepackage{ifsym}
\usepackage{color}
\usepackage{hyperref}
\usepackage{dsfont}

\newcommand{\ble}{\begin{lemma}}
\newcommand{\ele}{\end{lemma}}

\newtheorem{lemma}{Lemma}[section]

\newtheorem{theorem}[lemma]{Theorem}

\newtheorem{corollary}[lemma]{Corollary}

\newtheorem{observation}[lemma]{Observation}

\newtheorem{remark}[lemma]{Remark}

\newtheorem{claim}[lemma]{Claim}

\newtheorem{algorithm}[lemma]{Algorithm}

\theoremstyle{definition}
\newtheorem{definition}[lemma]{Definition}

\newcommand{\beao}{\begin{eqnarray*}}
\newcommand{\eeao}{\end{eqnarray*}\noindent}
\newcommand{\beam}{\begin{eqnarray}}
\newcommand{\eeam}{\end{eqnarray}\noindent}

\newcommand{\lp}[2][p]{\|#2\|_{#1}}

\newcommand{\nats}{{\mathbb N}}
\newcommand{\eps}{{\varepsilon}}

\newcounter{note}[section]

\begin{document}

\title{The Norms of Graph Spanners}
\author{Eden Chlamt\'a\v{c}\footnote{Supported in part by ISF grant 1002/14}\\Ben Gurion University \and Michael Dinitz\footnote{Supported in part by NSF awards CCF-1464239 and CCF-1535887}\\Johns Hopkins University \and Thomas Robinson\footnote{Supported in part by ISF grant 1002/14}\\Ben Gurion University}

\maketitle

\begin{abstract}
    A $t$-spanner of a graph $G$ is a subgraph $H$ in which all distances are preserved up to a multiplicative $t$ factor.  A classical result of Alth\"ofer et al.~is that for every integer $k$ and every graph $G$, there is a $(2k-1)$-spanner of $G$ with at most $O(n^{1+1/k})$ edges.  But for some settings the more interesting notion is not the number of edges, but the degrees of the nodes.  This spurred interest in and study of spanners with small maximum degree.  However, this is not necessarily a robust enough objective: we would like spanners that not only have small maximum degree, but also have ``few'' nodes of ``large'' degree.  To interpolate between these two extremes, in this paper we initiate the study of graph spanners with respect to the $\ell_p$-norm of their degree vector, thus simultaneously modeling the number of edges (the $\ell_1$-norm) and the maximum degree (the $\ell_{\infty}$-norm).  We give 
    precise upper bounds for all ranges of $p$ and stretch $t$: we prove that the greedy $(2k-1)$-spanner has $\ell_p$ norm of at most $\max(O(n), O(n^{\frac{k+p}{kp}}))$, and that this bound is tight (assuming the Erd\H{o}s girth conjecture).  We also study universal lower bounds, allowing us to give ``generic'' guarantees on the approximation ratio of the greedy algorithm which generalize and interpolate between the known approximations for the $\ell_1$ and $\ell_{\infty}$ norm.  Finally, we show that at least in some situations, the $\ell_p$ norm behaves fundamentally differently from $\ell_1$ or $\ell_{\infty}$: there are regimes ($p=2$ and stretch $3$ in particular) where the greedy spanner has a provably superior approximation to the generic guarantee.
\end{abstract}

\section{Introduction}

Graph spanners are subgraphs which approximately preserve distances.  Slightly more formally, given a graph $G = (V, E)$ (possibly with lengths on the edges), a subgraph $H$ of $G$ is a \emph{$t$-spanner} of $G$ if $d_G(u,v) \leq d_H(u,v) \leq t \cdot d_G(u,v)$ for all $u,v \in V$, where $d_G$ denotes shortest-path distances in $G$ (and $d_H$ in $H$).  The value $t$ is called the \emph{stretch} of the spanner.

Graph spanners were originally introduced in the context of distributed computing~\cite{PU89,PS89}, but have since proved to be a fundamental building block that is useful in a variety of applications, from property testing~\cite{BGJRW09} to network routing~\cite{TZ01}.  When building spanners there are many objectives which we could try to optimize, but probably the most popular is the number of edges (the \emph{size} or the \emph{sparsity}).  Not only is sparsity important in many applications, it also admits a beautiful tradeoff with the stretch, proved by Alth\"ofer et al.~\cite{ADDJS93}:

\begin{theorem}[\cite{ADDJS93}] \label{thm:ADDJS}
For every integer $k \geq 1$ and every weighted graph $G = (V, E)$ with $|V| = n$, there is a $(2k-1)$-spanner $H$ of $G$ with at most $O(n^{1+1/k})$ edges.
\end{theorem}


While understanding the tradeoff between the size and the stretch was a seminal achievement, for many applications (particularly in distributed computing) we care not just about the size, but also about the \emph{maximum degree}. 
Unfortunately, unlike the size, there is no possible tradeoff between the stretch and the maximum degree.  This is trivial to see: if $G$ is a star,  then the only spanner of $G$ with non-infinite stretch has maximum degree of $n-1$.  In general, if $G$ has maximum degree $\Delta$, then all we can say is the trivial fact that $G$ has a spanner with maximum degree at most $\Delta$.  Nevertheless, given the importance of the maximum degree objective, there has been significant work on building spanners that minimize the maximum degree from the perspective of \emph{approximation algorithms}~\cite{KP98,CDK12,CD14}.  From this perspective, we are given a graph $G$ and stretch value $t$ and are asked to find the ``best'' $t$-spanner of $G$ (where ``best'' means minimizing the maximum degree).  

While this has been an interesting and productive line of research, clearly there are problems with the maximum degree objective as well.  For example, if it is unavoidable for there to be some node of large degree $d$, the maximum degree objective allows us to make every other vertex also of degree $d$, with no change in the objective function.  But clearly we would prefer to have fewer high-degree nodes if possible!

So we are left with a natural question: can we define a notion of ``cost'' of a spanner which discourages very high degree nodes, but if there are high degree nodes, still encourages the rest of the nodes to have small degree?  There is of course an obvious candidate for such a cost function: the $\ell_p$ norm of the degree vector.  That is, given a spanner $H$, we can define $\|H\|_p$ to be the $\ell_p$-norm of the $n$-dimensional vector in which the coordinate corresponding to a node $v$ contains the degree of $v$ in $H$.  Then $\|H\|_1$ is just (twice) the total number of edges, and $\|H\|_{\infty}$ is precisely the maximum degree.  Thus the $\ell_p$-norm is an interpolation between these two classical objectives.  Moreover, for $1 < p < \infty$, this notion of cost has precisely the properties that we want: it encourages low-degree nodes rather than high-degree nodes, but if high-degree nodes are unavoidable it still encourages the rest of the nodes to be as low-degree as possible.  These properties, of interpolating between the average and the maximum, are why the $\ell_p$-norm has appeared as a popular objective for a variety of problems, ranging from clustering (the famous $k$-means problem~\cite{kmeans1,kmeans2}), to scheduling~\cite{scheduling1,scheduling2,scheduling3}, to covering~\cite{covering1}.

\subsection{Our Results and Techniques} \label{sec:results}

In this paper we initiate the study of graph spanners under the $\ell_p$-norm objective.  We prove a variety of results, giving upper bounds, lower bounds, and approximation guarantees.  Our main result is the analog of Theorem~\ref{thm:ADDJS} for the $\ell_p$-norm objective, but we also characterize universal lower bounds as part of an effort to understand the generic approximation ratio for the related optimization problem.  We also show that in some ways the $\ell_p$-norm can behave fundamentally differently than the traditional $\ell_1$ or $\ell_{\infty}$ norms, by proving that the greedy algorithm can have an approximation ratio that is \emph{strictly better} than the generic guarantee, unlike the $\ell_1$ or $\ell_{\infty}$ settings.

\subsubsection{Upper Bound}

We begin by proving our main result: a universal upper bound (the analog of Theorem~\ref{thm:ADDJS}) for $\ell_p$-norm spanners.  Recall the classical greedy algorithm for constructing a $t$-spanner $H$ of a graph $G = (V, E)$. Consider the edges in nondecreasing order of edge length, and when considering edge $\{u,v\}$, add it to $H$ if currently $d_H(u,v) > t \cdot d_G(u,v)$.  We call $H$ the \emph{greedy $t$-spanner} of $G$.  It is trivial to show that the greedy $t$-spanner has girth at least $t+2$.  This is the algorithm that was used to prove Theorem~\ref{thm:ADDJS}, and it has since received extensive study (see, e.g., \cite{FS16,CDNS92}) and will form the basis of our upper bound:   

\begin{theorem} \label{thm:upper-main}
Let $k\geq 1$ be an integer, let $G = (V, E)$ be a graph (possibly with lengths on the edges), and let $H = (V, E_H)$ be the greedy $(2k-1)$-spanner of $G$.  Then $\lp{H} \leq \max\left(O(n), O\left(n^{\frac{k+p}{kp}}\right)\right)$ for all $p \geq 1$.
\end{theorem}

In other words, if $p \geq \frac{k}{k-1}$ then our upper bound is $O(n)$, and otherwise it is $O\left(n^{\frac{k+p}{kp}}\right)$.  Clearly this interpolates between $p=1$ and $p = \infty$: when $p=1$ this is the same bound as Theorem~\ref{thm:ADDJS}, while if $p = \infty$ this gives $O(n)$ which is the only possible bound in terms of $n$.  It is also straightforward to prove that this bound is tight if we again assume the Erd\H{o}s girth conjecture~\cite{Erdos-girth}; for completeness, we do this in Appendix~\ref{app:UB-tight}. 

The proof of Theorem~\ref{thm:ADDJS} from~\cite{ADDJS93} is relatively simple: the greedy $(2k-1)$-spanner has girth at least $2k+1$, and any graph with more than $n^{1+1/k}$ edges must have a cycle of length at most $2k$.  Generalizing this to the $\ell_p$-norm is significantly more complicated, since it is not nearly as easy to show a relationship between the girth and the $\ell_p$-norm.  But this is precisely what we do.  

It turns out to be easiest to prove Theorem~\ref{thm:upper-main} for stretch $3$: it just takes one more step beyond~\cite{ADDJS93} to split the vertices of the high-girth graph (the spanner) into ``low'' and ``high'' degrees, and show that each vertex set does not contribute too much to the $\ell_p$ norm.  However for larger stretch values this approach does not work: the main lemma used for stretch $3$ (Lemma~\ref{lem:stretch3-large}) is simply false when generalized to larger stretch bounds.  Instead, we need a much more involved decomposition into ``low'', ``medium'', and ``high''-degree nodes.  This decomposition is very subtle, since the categories are not purely about the degree, but rather about how the degree relates to expansion at some particular distances from the node.  We also need to further decompose the ``high''-degree nodes into sets determined by which distance level we consider the expansion of.  We then separately bound the contribution to the $p$-norm of each class in the decomposition; for ``low''-degree nodes this is quite straightforward, 
but for medium and high-degree nodes this requires some subtle arguments which strongly use the structure of large-girth graphs.  

\subsubsection{Universal Lower Bounds}

To motivate our next set of results, consider the optimization problem of finding the ``best'' $t$-spanner of a given input graph.  When ``best'' is the smallest $\ell_1$-norm this is known as the \textsc{Basic $t$-Spanner} problem~\cite{DK11-stoc,BBMRY13,DKR16,DZ16}, and when ``best'' is the smallest $\ell_{\infty}$-norm this is the \textsc{Lowest-Degree $t$-Spanner} problem~\cite{KP98,CDK12,CD14}.  It is natural to consider this problem for the $\ell_p$-norm as well.  It is also natural to consider how well the greedy algorithm (used to prove the upper bound of Theorem~\ref{thm:upper-main}) performs as an approximation algorithm.   

To see an obvious way of analyzing the greedy algorithm as an approximation algorithm, consider the $\ell_1$-norm.  Theorem~\ref{thm:ADDJS} implies that the greedy algorithm always returns a spanner of size at most $O(n^{1+1/k})$, while clearly \emph{every} spanner must have size at least $\Omega(n)$ (assuming that the input graph is connected).  Thus we immediately get that the greedy algorithm is an $O(n^{1/k})$-approximation.  By dividing a universal upper bound (an upper bound on the size of the greedy spanner that holds for every graph) by a universal lower bound (a lower bound on the size of \emph{every} spanner in every graph), we can bound the approximation ratio in a way that is generic, i.e., that is essentially independent of the actual graph.

Now consider the $\ell_{\infty}$-norm.  The generic approach seems to break down here: the universal upper bound is only $\Theta(n)$ (as shown by the star graph), while the universal lower bound is only $\Theta(1)$ (as shown by the path).  So it seems like the generic guarantee is just the trivial $\Theta(n)$.  But this is just because $n$ is the wrong parameter in this setting: the correct parameterization is based on $\Delta$, the maximum degree of $G$ (i.e., $\Delta=\|G\|_\infty$). With respect to $\Delta$, the greedy algorithm (or any algorithm) returns a spanner with maximum degree at most $\Delta$, while \emph{any} $t$-spanner of a graph with maximum degree $\Delta$ must have maximum degree at least $\Omega(\Delta^{1/t})$ (assuming the graph is unweighted).  So there is still a ``generic'' guarantee which implies that the greedy algorithm is an $O(\Delta^{1-1/t}) \leq O(n^{1-1/t})$-approximation.

This suggests that for $1 < p < \infty$, we will need to parameterize by \emph{both} the number of nodes $n$ and the $\ell_p$-norm $\Lambda$ of $G$.  We can define both universal upper bounds and universal lower bounds with respect to this dual parameterization: 

\begin{align*}
\mathrm{UB}_t^p(n,\Lambda)&=\max_{\substack{G=(V, E): |V|=n, \|G\|_p = \Lambda,\\ \text{and $G$ is connected}}} \;\min_{H:\text{ $H$ is a $t$-spanner of $G$}}\|H\|_p\\
\mathrm{LB}_t^p(n,\Lambda)&=\min_{\substack{G=(V, E): |V|=n, \|G\|_p = \Lambda,\\ \text{and $G$ is connected}}} \;\min_{H:\text{ $H$ is a $t$-spanner of $G$}}\|H\|_p 
\end{align*}

With this notation, we can define the generic guarantee $g_t^p(n, \Lambda) = \mathrm{UB}_{t}^p(n, \Lambda) / \mathrm{LB}_{t}^p(n, \Lambda)$, and if we want a guarantee purely in terms of $n$ we can define the generic guarantee $g_t^p(n) = \max_{\Lambda} g_t^p(n, \Lambda)$.  Our upper bound of Theorem~\ref{thm:upper-main} can then be restated as the claim that $$\mathrm{UB}_{2k-1}^p(n, \Lambda) \leq \min\left\{\Lambda, \max\left\{O(n), O(n^{\frac{k+p}{kp}})\right\}\right\}$$ for all $n, k, p, \Lambda$.  So in order to understand the  generic guarantees $g_{2k-1}^p(n, \Lambda)$ or $g_{2k-1}^p(n)$, we need to understand the universal lower bound quantity $\mathrm{LB}_{2k-1}^p(n,\Lambda)$. 


Surprisingly, unlike the $\ell_1$ and $\ell_{\infty}$ cases, the universal lower bound for other values of $p$ is extremely complex.  Understanding its value, and understanding the structure of the extremal graphs which match the bound given by $\mathrm{LB}_t^p(n,\Lambda)$, are the most technically involved results in this paper.  However, while the analysis and even the exact formulation of the lower bound is quite complex, it turns out to be easily computable from a simple linear program:
\begin{theorem} There is an explicit linear program of size $O(t)$ which calculates $\mathrm{LB}_{t}^p(n, \Lambda)$ for any $t\in\nats, p\geq 1$. The bound given by the program is tight up to a factor of $\log(n)^{O(t)}$.
\label{thm:lower-meta}
\end{theorem}

Our linear program and the proof of Theorem~\ref{thm:lower-meta} appear in Section~\ref{sec:LP}. In fact, our linear program not only calculates a lower bound on the $\ell_p$-norm of any $t$-spanner, it also gives the parameters which define an extremal graph of $\ell_p$-norm $\Lambda$ with a $t$-spanner whose $\ell_p$-norm matches this lower bound. While the structure of these extremal graphs is simple, the dependence of the parameters of these graphs on $t$ and $p$ is quite complex. Nevertheless, we give a complete explicit description of these graphs for every possible value of $p,t$. 

Interestingly, despite the fact that $\mathrm{LB}_{t}^p(n, \Lambda)$ is fundamentally a question of extremal graph theory  (although as discussed our motivation is the generic guarantee on approximation algorithms), our techniques are in some ways more related to approximation algorithms.  We give a linear program which computes the LB function, and we reason about it by explicitly constructing dual solutions.  This is, to the best of our knowledge, the first time that structural bounds on spanners (as opposed to approximation bounds) have been derived using linear programs.  Moreover, the structure of the extremal graphs is fundamentally related to a quantity which we call the \emph{$p$-log density} of the input graph.  This is a generalization of the notion of ``log-density'', which was introduced as the fundamental parameter when designing approximation algorithms for the \textsc{Densest $k$-Subgraph} (DkS) problem~\cite{BCCFV10}, and has since proved useful in many approximation settings (see, e.g., \cite{CDK12,CDM17,CMMV17,CM18}).

\subsubsection{Greedy Can Do Better Than The Generic Bound}



As discussed, when $p=1$ or $p=\infty$, the approximation ratio of the greedy algorithm can be bounded by the generic guarantee.  But it turns out that the connection is actually even closer: when $p=1$ and $p=\infty$, for every $n$ and $\Lambda$ the approximation ratio of the greedy algorithm is \emph{equal} to the generic guarantee $g_{2k-1}^p(n, \Lambda)$.  In other words, greedy is no better than generic in the traditional settings (we prove this for completeness, but it is essentially folklore).  In fact, for the $\ell_1$ objective, giving \emph{any} approximation algorithm which is better than the generic guarantee $g_{2k-1}^1(n)$ is a long-standing open problem~\cite{DZ16} which has only been accomplished for stretch $3$~\cite{BBMRY13} and stretch $4$~\cite{DZ16}, while for the $\ell_{\infty}$ objective such an improvement was only shown recently~\cite{CD14} (and not with the greedy algorithm).  



We show that, at least in some regimes of interest, $\ell_p$-norm graph spanners exhibit fundamentally different behavior from $\ell_1$ and $\ell_{\infty}$: the greedy algorithm has approximation ratio which is \emph{better} than the generic guarantee, even though the universal upper bound is proved via the greedy algorithm!  In particular, we consider the regime of stretch $3$, $p=2$, and $\Lambda = \Theta(n)$.  This is a very natural regime, since $p=2$ is the most obvious and widely-studied norm other than $\ell_1$ and $\ell_{\infty}$, and stretch $3$ is the smallest value for which nontrivial sparsification can occur.  

Our theorems about UB and LB imply that $g_3^2(n) = g_3^2(n,n) = \Theta(\sqrt{n})$.  But we show that in this setting (and in fact for any $\Lambda$ as long as $p=2$ and the stretch is $3$) the greedy algorithm is an $O(n^{63/128})$-approximation.  
Thus we show that, unlike $\ell_1$ and $\ell_{\infty}$, for $p=2$ the greedy algorithm provides an approximation guarantee that is strictly better than the generic bound, both for specific values of $\Lambda$ and when considering the worst case $\Lambda$.  

\subsection{Outline}
We begin in Section~\ref{sec:definitions} with some basic definitions and preliminaries.  In order to illustrate the basic concepts in a simpler and more understandable setting, we then focus in Section~\ref{sec:stretch3} on the special case of stretch $3$: we prove the stretch-$3$ version of Theorem~\ref{thm:ADDJS} in Section~\ref{sec:stretch3-upper}, and then show that the greedy algorithm has approximation ratio better than the generic guarantee in Section~\ref{sec:greedy-generic}.  We then prove our upper and lower bounds in full generality: the upper bound (i.e., the proof of Theorem~\ref{thm:upper-main}) in Section~\ref{sec:upper}, and then our universal lower bound in Section~\ref{sec:lower-bound-main}.  Due to space constraints, all missing proofs can be found in the appendices.

\section{Definitions and Preliminaries} \label{sec:definitions}

Let $G = (V, E)$ be a graph, possibly with lengths on the edges.  For any vertex $u \in V$, we let $d(u)$ denote the degree of $u$ and let $N(u)$ denote the neighbors of $u$.  We will also generalize this notation slightly by letting $N_i(u)$ denote the set of vertices that are exactly $i$ hops away from $u$ (i.e., their distance from $u$ if we ignore lengths is exactly $i$), and we let $d_i(u) = |N_i(u)|$.  Note that by definition, $N_0(u) = \{u\}$ and $d_0(u) = 1$ for all $u \in V$.  We will sometimes use $B(v, r) = \cup_{i=0}^i N_i(v)$ to denote the ball around $v$ of radius $r$.

We let $d_G : V \times V \rightarrow \mathbb{R}_{\geq 0}$ denote the shortest-path distances in $G$.  A subgraph $H = (V, E_H)$ of a graph $G = (V, E)$ is a \emph{$t$-spanner} of $G$ if $d_H(u,v) \leq t \cdot d_G(u,v)$ for all $u,v \in E$.  Recall that $\|\vec{x}\|_p = \left(\sum_{i=1}^n x_i^p\right)^{1/p}$ for any $p \geq 1$ and $\vec{x} \in \mathbb{R}^n$.   To measure the ``cost'' of a spanner, for any graph $G = (V, E)$, let $\vec{d_G}$ denote the vector of degrees in $G$ and for any $p \geq 1$, let $\lp{G} = \lp{\vec{d_G}}$.  For any subset $S \subseteq V$, we let $\|S\|_p$ denote the $\ell_p$ norm of the vector obtained from $\vec{d_G}$ by removing the coordinate of every node not in $S$ (note that we do not remove the nodes from the graph, i.e., $\|S\|_p$ is the norm of the degrees in $G$ of the nodes in $S$, not in the subgraph induced by $S$).

\section{Warmup: Stretch $3$} \label{sec:stretch3}
We begin by analyzing the special case of stretch $3$, particularly for the $\ell_2$-norm.  More specifically, we will focus on bounding $\mathrm{UB}_3^p(n,\Lambda)$.  This is one of the simplest cases, but demonstrates (at a very high level) the outlines of our upper bound.  Moreover, in this particular case we can prove that the greedy algorithm performs better than the generic guarantee, showing a fundamental difference between the $\ell_2$ norm and the more traditional $\ell_1$ and $\ell_{\infty}$ norms.  

\subsection{Upper Bound} \label{sec:stretch3-upper}
Recall that \emph{greedy spanner} is the spanner obtained from the obvious greedy algorithm: starting with an empty graph as the spanner, consider the edges one at a time in nondecreasing length order, and add an edge if the current spanner does not span it (within the given stretch requirement).  It is obvious that when run with stretch parameter $t$ this algorithm does indeed return a $t$-spanner, and moreover it will return a $t$-spanner that has girth at least $t+2$ (if there is a $(t+1)$-cycle then the algorithm would not have added the final edge).  

Our main goal in this section will be to prove the following theorem

\begin{theorem} \label{thm:upper-3}
Let $G = (V, E)$ be a graph and let $H = (V, E_H)$ be the greedy $3$-spanner of $G$.  Then $\lp{H} \leq \max(O(n), O(n^{(2+p)/(2p)}))$ for all $p \geq 1$.
\end{theorem} 

In other words, when $1 \leq p \leq 2$ the greedy $3$-spanner $H$ has $\lp{H} \leq O(n^{(2+p)/(2p)})$, and when $p \geq 2$ we get that that $\lp{H} \leq O(n)$.

To prove this theorem, we will use first show that nodes with ``large'' degree cannot be incident on too many edges in any graph of girth at least $5$ (like the greedy $3$-spanner).  This is the most important step, since for $p > 1$ the $p$-norm of a graph gives greater ``weight'' to nodes with larger degree.  

\begin{lemma} \label{lem:stretch3-large}
Let $G=(V,E)$ be a graph with girth at least 5. Then  $\sum_{v\in V:d(v)\geq 2\sqrt{n}}d(v)\leq 2n.$
\end{lemma}
\begin{proof} 
Suppose for the sake of contradiction that these vertices have total degree greater than $2n$, and let $\{v_1,\ldots,v_{\ell+1}\}$ be a minimal set with this property. That is, all these vertices have degree at least $2\sqrt{n}$, and furthermore $\sum_{i=1}^{\ell}d(v_i)\leq 2n < \sum_{i=1}^{\ell+1}d(v)$.

Because $G$ has girth at least 5, any two vertices $v_i,v_j$ in this set have at most one common neighbor. That is, $|N(v_i)\cap N(v_j)|\leq 1$. Thus, for every $j\in[\ell+1]$, the number of ``new'' neighbors contributed by $N(v_j)$ is $\left|N(v_j)\setminus\left(\bigcup_{i=1}^{j-1}N(V_i)\right)\right|\geq |N(v_j)|-\sum_{i=1}^{j-1}|N(v_i)\cap N(v_j)|\geq d(v_j)-(j-1)\geq d(v_j)-\ell$.

On the other hand, we have $2n\geq\sum_{i=1}^\ell d(v_i)\geq \ell\cdot 2\sqrt{n}$, and so we have $\ell\leq \sqrt{n}$. Thus, every $v_j$ contributes at least $d(v_j)-\ell\geq d(v_j)-\sqrt{n}\geq d(v_j)/2$ new neighbors, and so we get $n\geq\left|\bigcup_{j=1}^{\ell+1}N(v_j)\right|=\sum_{j=1}^{\ell+1}\left|N(v_j)\setminus\left(\bigcup_{i=1}^{j-1}N(v_i)\right)\right|\geq\sum_{j=1}^{\ell+1}d(v_j)/2$, which contradicts our assumption that $\sum_{j=1}^{\ell+1}d(v_j)>2n$.
\end{proof}

We can now prove Theorem~\ref{thm:upper-3}.

\begin{proof}[Proof of Theorem~\ref{thm:upper-3}]
Let $V_{low} = \{v \in V : d(v) \leq 2\sqrt{n}\}$, and let $V_{high} = \{v \in V : d(v) > 2\sqrt{n}\}$.  Since $H$ has girth at least $5$, we can apply Lemma~\ref{lem:stretch3-large}.  So using this lemma and standard algebraic inequalities, we get that
\begin{align*}
    \lp{H} &= \left(\sum_{v \in V_{low}} d(v)^p + \sum_{v \in V_{high}} d(v)^p \right)^{1/p}
    \leq \left(\sum_{v \in V_{low}} d(v)^p\right)^{1/p} + \left(\sum_{v \in V_{high}} d(v)^p\right)^{1/p} \\
    &\leq \left(\sum_{v \in V_{low}} (2\sqrt{n})^p\right)^{1/p} + \sum_{v \in V_{high}} d(v)
    \leq \left(n \cdot 2 n^{p/2}\right)^{1/p} + \sum_{v \in V_{high}} d(v) 
    \leq n^{\frac{2+p}{2p}} + 2n,
\end{align*}
which implies the theorem.
\end{proof}



It is easy to show that the above bound is tight: for every $p\geq 1$ there are graphs in which every $3$-spanner has size at least $\max(\Omega(n), \Omega(n^{\frac{2+p}{2p}}))$.  In fact, we can generalize slightly to also account for different values of $\Lambda$.  Theorem~\ref{thm:upper-main} can be interpreted as claiming that $\mathrm{UB}_3^p(n, \Lambda) \leq O(\min(\max(n, n^{\frac{2+p}{2p}}), \Lambda))$.  
In Appendix~\ref{app:UB-tight} we show (Theorem~\ref{thm:UB-tight}) that this is tight: $\mathrm{UB}_3^p(n, \Lambda) \geq \Omega(\min(\max(n, n^{\frac{2+p}{2p}}), \Lambda))$ for all $p \geq 1$ and $\Omega(n^{1/p}) \leq \Lambda \leq O(n^{\frac{1+p}{p}})$.

\subsection{Greedy vs Generic} \label{sec:greedy-generic}
It is not hard to show that in the traditional settings in which spanners have been studied, the $\ell_1$ and $\ell_{\infty}$ norms, the greedy algorithm does no better than the generic guarantee, for all relevant parameter regimes.  In slightly more detail, for $\ell_{\infty}$ it is relatively easy to show that $\mathrm{UB}_{t}^{\infty}(n, \Lambda) = \Theta(\Lambda)$, while $\mathrm{LB}_{t}^{\infty}(n, \Lambda) = \Theta(\Lambda^{1/t})$.  Thus the generic guarantee $g_t^{\infty}(n, \Lambda) = \Theta(\Lambda^{1-\frac{1}{t}})$, and moreover we can build graphs in which the approximation ratio of the greedy algorithm is also $\Theta(\Lambda^{1-\frac{1}{t}})$.  Similarly, for the $\ell_1$-norm, classical results on spanners imply that $\mathrm{UB}_{2k-1}^{1}(n, \Lambda) = \Theta(\min(n^{1+\frac{1}{k}}), \Lambda))$ and $\mathrm{LB}_{2k-1}^{1}(n, \Lambda) = \Theta(n)$, so the generic guarantee is $g_{2k-1}^1(n, \Lambda) = \Theta(\min(n^{1+\frac{1}{k}}), \Lambda) / n)$ and there are graphs for all parameter regimes where this is the approximation ratio achieved by greedy.

We show that the behavior of the greedy spanner in intermediate $\ell_p$-norms is fundamentally different: in some parameter regimes of interest, greedy \emph{outperforms} the generic guarantee!

To demonstrate this, consider the regime of stretch $3$ with the $\ell_2$ norm and with $\Lambda = n$.  In this regime, the results of Section~\ref{sec:stretch3-upper} imply that $\mathrm{UB}_{3}^{2}(n, n) = \Theta(n)$.  On the other hand, our results on the universal lower bound from Section~\ref{sec:lower-bound-main} (Corollary~\ref{cor:LB-low-stretch} in particular) directly imply that $\mathrm{LB}_{3}^{2}(n, n) = \tilde \Theta(\sqrt{n})$.  Thus the generic guarantee is $g_3^2(n, n) = \tilde \Theta(\sqrt{n})$, and this is the worst case over $\Lambda$ and thus $g_3^2(n) = \tilde \Theta(\sqrt{n})$. However, we show that the greedy algorithm is a strictly better approximation, even without parameterizing by $\Lambda$.

\begin{theorem} \label{thm:greedy-approx}
The greedy algorithm is an $O(n^{63/128})$-approximation for the problem of computing $3$-spanner with smallest $\ell_2$-norm.  
\end{theorem}

To prove this, let $G = (V, E)$ be a graph with $|V| = n$, let $H$ be the greedy $3$-spanner of $G$, and let $H^*$ be the $3$-spanner of $G$ with minimum $\|H\|_2$.  Let $\alpha = \log_n \|H^*\|_2$, so $\|H^*\|_2 = n^{\alpha}$; note that $\alpha \geq 1/2$.  We first prove a lemma which uses $\|H^*\|_2$ to bound neighborhoods.

\begin{lemma} \label{lem:opt-neighborhood}
    $|B_{H^*}(v,r)| \leq n^{\left(2 -\frac{1}{2^{r-1}}\right)\alpha}$ for all $v \in V$ and $r \in \mathbb{N}$.
\end{lemma}
\begin{proof}
We use induction on $r$.  For the base case $r = 1$, since $\|H^*\|_2 = n^{\alpha}$ we know that $v$ has degree at most $n^{\alpha}$, and thus $|B_{H^*}(v,1)| \leq n^{\alpha} = n^{\left(2 - \frac{1}{2^{r-1}}\right)\alpha}$.

Now suppose that the theorem is true for some integer $r$.  Let $|B_{H^*}(v,r)| = n^{\gamma} \leq n^{\left(2 - \frac{1}{2^{r-2}}\right)\alpha}$ (by induction).  Since $\|H^*\|_2 = n^{\alpha}$, the average degree (in $H^*$) of the nodes in $B_{H^*}(v,r)$ is at most $n^{\alpha - (\gamma / 2)}$.  Thus we get that
$|B_{H^*}(v, r+1)| \leq n^{\gamma} \cdot n^{\alpha - (\gamma / 2)} = n^{\alpha + (\gamma / 2)} \leq n^{\alpha + \left(1 - \frac{1}{2^{r-1}}\right)\alpha} = n^{\left(2 - \frac{1}{2^{r-1}}\right)\alpha}$, as claimed.
\end{proof}

Using this lemma, we can now prove Theorem~\ref{thm:greedy-approx}.

\begin{proof}[Proof of Theorem~\ref{thm:greedy-approx}]
Lemma~\ref{lem:opt-neighborhood} implies that $|B_{H^*}(v, 6)| \leq n^{(63/32)\alpha}$ for all $v \in V$.  Since $H^*$ is a $3$-spanner of $G$, every vertex in $B_G(v,2)$ must be in $B_{H^*}(v,6)$, and thus $|B_G(v,2)| \leq n^{(63/32)\alpha}$.  Now we can use this to bound the number of $2$-paths in $H$.  Let $P_2(H)$ denote the number of paths of length $2$ in $H$.  Since $H$ is the greedy $3$-spanner of $G$ it must have girth at least $5$.  This means that every path of length $2$ in $H$ which starts from $v$ must have a different other endpoint: there cannot be two different paths of the form $v - w - u$ and $v - x - u$ in $H$, or else $H$ would have girth at most $4$.  Thus the number of $2$-paths in $H$ which start from $v$ is bounded by $|B_H(v, 2)| \leq |B_G(v,2)| \leq n^{(63/32)\alpha}$, and thus $P_2(H) \leq n^{1 + (63/32)\alpha}$.

On the other hand, note that instead of counting $2$-paths in $H$ by their starting vertex, we could instead count them by their middle vertex.  The number of $2$-paths where $v$ is the \emph{middle} node is $d_H(v)^2$, and thus $P_2(H) = \sum_{v \in V} d_H(v)^2 = \|H\|_2^2$.  Combining these two inequalities implies that $\|H\|_2 \leq n^{\frac12 + \frac{63}{64}\alpha}$, and hence the greedy spanner has approximation ratio of at most
$\frac{\|H\|_2}{\|H^*\|_2} \leq \frac{n^{\frac12 + \frac{63}{64}\alpha}}{n^{\alpha}} = n^{\frac12 - \frac{1}{64}\alpha} \leq n^{\frac12 - \frac{1}{128}} = n^{63 / 128}$.
\end{proof}

\section{Upper Bound: General Stretch} \label{sec:upper}

We now want to generalize the bounds from Section~\ref{sec:stretch3} to hold for larger stretch ($2k-1$ in particular) in order to prove Theorem~\ref{thm:upper-main}.  A natural approach would be an extension of the stretch $3$ analysis: if in Lemma~\ref{lem:stretch3-large} we replaced the the bound of $2\sqrt{n}$ with $2 n^{1/k}$, then the proof of Theorem~\ref{thm:upper-3} could easily be extended to prove Theorem~\ref{thm:upper-main}.  Unfortunately this is impossible: there are graphs of girth at least $2k+1$ where it is not true that the number of edges incident on nodes of degree at least $2n^{1/k}$ is at most $O(n)$.  This can be seen from, e.g., \cite{N01} for $k=3$.  

So we cannot just break the vertices into ``high-degree'' and ``low-degree'' as we did for stretch $3$.  Instead, our decomposition is more complicated.  We will still have low-degree nodes, which can be analyzed trivially.  But our definition of ``high'' will actually be parameterized by a distance $j$, and we will define a node to be ``high-degree'' at distance $j$ if its degree is large relative to the expansion of its neighborhood at approximately distance $j$.  We will also introduce a new type of ``medium-degree'' node. In Section~\ref{sec:upper-decomposition} we define this decomposition and prove that it is a full decomposition of $V$, and then in Sections~\ref{sec:upper-lemmas} and \ref{sec:upper-final} we show that no part in this decomposition can contribute too much to the overall cost.

First, though, we make one simple observation that will allow us to simplify notation by only considering one particular value of $p$.  While we could analyze general values of $p$ as we did for stretch $3$ in Section~\ref{sec:stretch3-upper}, it is actually sufficient to prove the bound for the special case of $k$ and $p$ where the two terms in the maximum are equal, i.e., when  $\frac{k+p}{kp} = 1$.  The following is a straightforward application of H\"older's inequality.


\begin{lemma} \label{lem:upper-main-reduction}
Let $k \geq 1$ be an integer, let $G = (V, E)$ be a graph, and let $H = (V, E_H)$ be the greedy $(2k-1)$-spanner of $G$.  If $\|H\|_{p'} = O (n)$ for $p'=k/(k-1)$ then
$\|H\|_{p} \leq \max\left(O(n), O\left(n^{\frac{k+p}{kp}}\right)\right)$ for all $p \geq 1$.
\end{lemma}
\begin{proof}
First note that $p'=k/(k-1)$ if and only if $\frac{k+p'}{kp'}=1$.  So we break into two cases, one for $p > p'$ and one for $1 \leq p< p'$.  For the first case, where $p> p'$, the result follows simply because of the monotonicity of $p$-norms: $\|H\|_p \leq \|H\|_{p'} = O(n) = \max\left(O(n), O\left(n^{\frac{k+p}{kp}}\right)\right)$.

For the second case, where $1 \leq p < p'$, let $q$ be the value such that  $1 \leq p \leq p'$ and $\frac{1}{p'}+\frac{1}{q}=\frac{1}{p}$.  Recall that ${\vec{d_H}}$ is the degree vector of $H$.  Then H\"older's inequality implies that 
$\|{\vec{d_H}}\|_{p} \leq \|{\vec{d_H}}\|_{p'} \|1\|_{q} = n^{\frac{1}{p}-\frac{1}{p'}} \|{\vec{d_H}}\|_{p'}$.
Since by assumption we have $\|\vec{d_H}\|_{p'} \leq O(n)$, this implies that $\|H\|_p \leq O\left(n^{1+\frac{1}{p} - \frac{1}{p'}}\right) = O\left(n^{\frac{1}{p}-\frac{k-1}{k}+1}\right) = O\left(n^{\frac{k+p}{kp}}\right)$, as claimed.  
\end{proof}




\subsection{Graph Decomposition} \label{sec:upper-decomposition}

Recall that $d_{i}(v)$ denotes the number of vertices at distance exactly $i$ from $v$.  This will let us define the following vertex sets.  

\begin{definition}
Let $G=(V,E)$ be a graph of girth at least $2k+1$, with $k \geq 3$.  Then define 
\begin{align*}
V_{low} & = \{v \in V : d_1(v) \leq
n^{1/k}\} \\
V_{med} & = \{v \in V : n^{(k-2)/(k-1)} d_1(v)^{1/(k-1)} \leq d_{k-1}(v)\} \\
V_{high,j}&= \{v \in V : d_{k-2j-1}(v) \leq n^{1/(k-1)} d_{k-2j-3}(v) d_1(v)^{(k-2)/(k-1)}\},
\end{align*}
where $0 \leq j \leq \lfloor (k-3)/2 \rfloor$.
\end{definition}

It is not hard to see that this notion of high still corresponds to a deviation from regularity, as in the stretch $3$ setting; the difference is that this deviation is relative to the size of the neighborhood at distance $k-2j-1$ vs the neighborhood at distance $k-2j-3$.

As we will see in Sections~\ref{sec:upper-lemmas} and \ref{sec:upper-final}, analyzing the contribution of $V_{high,j}$ to the $p$-norm of the greedy spanner is in some sense the ``main'' technical step: analyzing $V_{low}$ is straightforward, and analyzing $V_{med}$, while nontrivial, turns out to be easier than the case for $V_{high,j}$.  Before we do this, though, we will show that we have a full decomposition of $V$.

\begin{theorem}
\label{thm:decomposition}
Let $G=(V,E)$ be a graph of girth at least $2k+1$, with $k \geq 3$.  Then $V=V_{low} \cup V_{med} \cup \left( \cup _{0 \leq j \leq \lfloor (k-3)/2 \rfloor} V_{high,j} \right)$.
\end{theorem}
\begin{proof}
We prove the case when $k$ is odd.  The other case is similar.  

Assume that $v \notin \cup_{0 \leq j \leq \lfloor (k-3)/2 \rfloor} V_{high,j}$.  Then by the definition of $V_{high, j}$, we know that $d_{k-2j-1}(v) > n^{1/(k-1)} d_{k-2j-3}(v) d_1(v)^{(k-2)/(k-1)}$ for all $j$.  Then a straightforward induction on $j$  implies that 
\begin{align}
\label{tel}
d_{k-1}(v) &> n^{1/2}d_1(v)^{(k-2)/2}.
\end{align}
If further we assume that $v \notin V_{low}$, then $d_1(v) > n^{1/k}$, and thus 
\begin{align}
\label{low}
    (d_1(v))^{k(k-3)/(2(k-1))} & \geq n^{(k-3)/(2(k-1))}.
\end{align}
Finally, assuming that $v \notin V_{med}$ implies that
\begin{align}
\label{med}
    n^{(k-2)/(k-1)} (d_1(v))^{1/(k-1)} &> d_{k-1}(v).
\end{align}
If we then multiply inequalities (\ref{tel}), (\ref{low}) and (\ref{med}), after some elementary algebra, we find $1 > 1$, which is a contradiction.  Thus $v \in V_{low} \cup V_{med} \cup \left( \cup _{0 \leq j \leq \lfloor (k-3)/2 \rfloor} V_{high,j} \right)$, implying the theorem.  
\end{proof}

\subsection{Structural Lemmas for High-Girth Graphs} \label{sec:upper-lemmas}

With Theorem~\ref{thm:decomposition} in hand, it remains to bound the contribution to the $p$-norm of the spanner of these different vertex sets.  In order to do this, we start with a few useful lemmas. We first prove a simple lemma: if the girth is large enough, then the neighborhoods around a node can be bounded by the neighborhoods around its neighbors.

\begin{lemma} \label{lem:backtrack}
Let $G = (V, E)$ have girth at least $2k+1$ with $k \geq 2$.  Then $\sum_{w \in N_1(v)} d_{k-1}(w) \leq d_k(v) + d_1(v) d_{k-2}(v)$ for all $v \in V$. 
\end{lemma}
\begin{proof}
Since $G$ has girth at least $2k+1$, for every $i \leq k$ and $w \in N_i(v)$ there is exactly one path of length $i$ from $v$ to $w$ (or else there would be a cycle of length at most $2k$).  Thus the $(k-1)$-neighborhoods of the neighbors of $v$ form a partition of $N_k(v)$ (when intersected with $N_k(v)$).  More formally, $N_k(v) = \cup_{w \in N_1(v)} \left(N_{k-1}(w) \cap N_k(v)\right)$, and $\left(N_{k-1}(w) \cap N_k(v)\right) \cap \left(N_{k-1}(u) \cap N_k(v)\right) = \emptyset$ for all $u \neq w \in N_1(v)$.  Moreover, the part of $N_{k-1}(w)$ which is not in $N_k(v)$ is a subset of $N_{k-2}(v)$, since the path from $w$ to any such node would go through $v$ as its first hop (where we consider $N_0(v) = \{v\}$).  Thus we get that $\sum_{w \in N_1(v)} d_{k-1}(w) = \sum_{w \in N_1(v)} \left(|N_{k-1}(w) \cap N_k(v)| + |N_{k-1}(w) \cap N_{k-2}(v)|\right) \leq |N_k(v)| + \sum_{w \in N_1(v)} |N_{k-2}(v)|= d_k(v) + d_1(v) d_{k-2}(v)$, as claimed.
\end{proof}

With this lemma in hand, we will now prove a more complicated technical lemma which will likewise hold for all high-girth graphs.  For a given $v,w$ with $v \in N(w)$, we can consider the fraction of the $k$-neighborhood of $w$ which is also contained in the $(k-1)$-neighborhood of $v$.  Then if we sum this fraction over all neighbors $v$ of $w$, we would of course get $1$ since the girth constraint would imply that any two neighbors of $v$ cannot both be first hops on paths to the same node in $N_k(w)$.  But what if we consider the slightly different ratio of $d_{k-1}(v) / d_k(w)$?  This is notably different since it includes in the numerator not just $N_{k-1}(v) \cap N_k(w)$, but also $N_{k-1}(v) \cap N_{k-2}(w)$.  It will prove useful for us to reason about these values, so we show that ``on average'' they behave approximately the same: if we sum up the neighbors of any given node then these fractions can add up to something quite large (not $1$), but overall they only add up to $O(n)$.

\begin{lemma} \label{lem:ratio}
Let $k \geq 1$ be an integer, and let $G = (V, E)$ have girth at least $2k+1$ and minimum degree at least $4$.  Then 
    $\sum_{w \in V} \sum_{v \in N(w)} \frac{d_{k-1}(v)}{d_{k}(w)} \leq 2n$.
\end{lemma}
\begin{proof}
For ease of notation, let $\Phi(k) = \sum_{w \in V} \sum_{v \in N(w)} \frac{d_{k-1}(v)}{d_{k}(w)}$.  We will prove that $\Phi(1) \leq n$ and that $\Phi(k) \leq n + \frac12 \Phi(k-1)$ for all $k \geq 2$.  These two statements clearly imply the lemma by a simple induction.

Let us first prove that $\Phi(1) \leq n$, which is the base case of the induction.  Starting from the definition of $\Phi(1)$, (and noting that $d_0(v) = 1$ by definition), we get that
\begin{align*}
    \Phi(1) &= \sum_{w \in V} \sum_{v \in N(w)} \frac{d_{0}(v)}{d_1(w)} \leq \sum_{w \in V} \frac{d_1(w)}{d_1(w)} = n.
\end{align*}

For $k > 1$, we can begin similarly, using the definition of $\Phi$ and now also Lemma~\ref{lem:backtrack} to get that
\begin{align*}
    \Phi(k) & = \sum_{w \in V} \sum_{v \in N(w)} \frac{d_{k-1}(v)}{d_{k}(w)} \leq \sum_{w \in V} \frac{d_{k}(w) + d_1(w) d_{k-2}(w)}{d_{k}(w)} = n + \sum_{w \in V} \frac{d_1(w) d_{k-2}(w)}{d_{k}(w)}.
\end{align*}

So now we need to prove that $\sum_{w \in V} \frac{d_1(w) d_{k-2}(w)}{d_{k}(w)} \leq \frac12 \Phi(k-1)$.  Let us first fix some $w \in V$ and try to lower bound $d_k(w)$.  Our assumption that every vertex has degree at least $4$ implies that $d_1(w) \geq \frac12 d_1(w) + 2$, and so $d_1(w) d_k(w) \geq \frac12 d_1(w) d_k(w) + 2 d_k(w)$.  This gives a lower bound on $d_k(w)$:
\begin{equation} \label{eq:min-deg1}
d_k(w) \geq \frac{\frac12 d_1(w) d_k(w) + 2 d_k(w)}{d_1(w)}.
\end{equation}
Now again using the fact that all vertices have degree at least $4$ (in fact, degree at least $3$ would be sufficient), and the fact that the girth is at least $2k+1$, we get a different lower bound: $d_k(w) \geq 2 d_{k-1}(w) \geq 4 d_{k-2}(w)$ for all $w \in V$.   Combining this with~\eqref{eq:min-deg1} gives us the bound
\begin{equation*}
    d_k(w) \geq \frac{2 d_1(w) d_{k-2}(w) + 2 d_k(w)}{d_1(w)} =2 \frac{d_1(w) d_{k-2}(w) + d_k(w)}{d_1(w)}.
\end{equation*}
Now we can apply Lemma~\ref{lem:backtrack} to the numerator, giving us
\begin{equation*}
    d_k(w) \geq 2\frac{\sum_{v \in N(w)}d_{k-1}(v)}{d_1(w)}.
\end{equation*}
The right hand side of this inequality is clearly (twice) the arithmetic mean of the values $\{d_{k-1}(v)\}_{v \in N(w)}$.  Since the arithmetic mean is at least the harmonic mean, we get that
\begin{equation} \label{eq:min-deg2}
    d_k(w) \geq 2\frac{\sum_{v \in N(w)}d_{k-1}(v)}{d_1(w)} \geq 2 \frac{d_1(w)}{\sum_{v \in N(w)} \frac{1}{d_{k-1}(v)}}.
\end{equation}

This is now finally the lower bound on $d_k(w)$ that we will use to prove that $\sum_{w \in V} \frac{d_1(w) d_{k-2}(w)}{d_{k}(w)} \leq \frac12 \Phi(k-1)$.  In particular, we immediately obtain
\begin{align*}
    \sum_{w \in V} \frac{d_1(w) d_{k-2}(w)}{d_{k}(w)} &\leq \frac12 \sum_{w \in V} \frac{d_1(w) d_{k-2}(w)}{\frac{d_1(w)}{\sum_{v \in N(w)} \frac{1}{d_{k-1}(v)}}} 
    = \frac12 \sum_{w \in V} d_{k-2}(w) \sum_{v \in N(w)} \frac{1}{d_{k-1}(v)} \\
    &= \frac12 \sum_{w \in V} \sum_{v \in N(w)} \frac{d_{k-2}(w)}{d_{k-1}(v)} 
    = \frac12 \sum_{w \in V} \sum_{v \in N(w)} \frac{d_{k-2}(v)}{d_{k-1}(w)} 
    = \frac12 \Phi(k-1).
\end{align*}
As shown, this implies the lemma.
\end{proof}

While Lemma~\ref{lem:ratio} is the main structural result that we will use to bound the ``high'' degree nodes, the following corollary makes it slightly simpler to use.
\begin{corollary} \label{techratio}
Let $k \geq2$ be an integer, and let $G=(V,E)$ have girth at least $2k+1$ and minimum degree at least $4$.  Then $\sum_{v \in V}\frac{(d_1(v))^2 d_{k-2}(v)}{d_k(v)+d_1(v)d_{k-2}(v)} \leq 2n$.
\end{corollary}
\begin{proof}
We have
\begin{align*}
\sum_{v \in V}\frac{(d_1(v))^2 d_{k-2}(v)}{d_k(v)+d_1(v)d_{k-2}(v)} 
&=\sum_{v \in V}\frac{(d_1(v)) d_{k-2}(v)}{\frac{d_k(v)+d_1(v)d_{k-2}(v)}{|N(v)|}}\\
&\leq \sum_{v \in V}\frac{(d_1(v)) d_{k-2}(v)}{\frac{\sum_{w \in N(v)}d_{k-1}}{|N(v)|}}. & (\text{Lemma}~\ref{lem:backtrack})
\end{align*}
Note the arithmetic mean in the denominator.  Some elementary algebra together with the fact that the arithmetic mean is at least the harmonic mean yields:
\begin{align*}
\sum_{v \in V}\frac{(d_1(v))^2 d_{k-2}(v)}{d_k(v)+d_1(v)d_{k-2}(v)} 
&\leq 
\sum_{v \in V}
\sum_{w \in N(v)}
\frac{d_{k-2}(v)}{d_{k-1}(w)}\\
&=
\sum_{w \in V}
\sum_{v \in N(w)}
\frac{d_{k-2}(v)}{d_{k-1}(w)}\\
&\leq 2n.  & (\text{Lemma}~\ref{lem:ratio}) 
\end{align*}
as claimed.
\end{proof}

\subsection{Proving Theorem~\ref{thm:upper-main}} \label{sec:upper-final}
We can now finally prove Theorem~\ref{thm:upper-main} by analyzing the contribution of the different sets in the decomposition to any graph of girth at least $2k+1$ (in particular, the greedy $(2k-1)$-spanner).  



The analysis of the low nodes is straightforward, while the analysis of the medium nodes is slightly more complex.  But the main difficulty is in the high nodes.
\begin{lemma} \label{lem:low}
Let $k \geq 2$ be an integer and let $G = (V, E)$ be a graph with girth at least $2k+1$.  Then $\|V_{low}\|_{\frac{k}{k-1}} \leq n$.
\end{lemma}
\begin{proof}
This is a straightforward calculation using only the definition of $V_{low}$: 
    \begin{equation*}
        \|V_{low}\|_{\frac{k}{k-1}} = \left(\sum_{v \in V_{low}} (d_1(v))^\frac{k}{k-1} \right)^{\frac{k-1}{k}} \leq \left( \sum_{v \in V_{low}} n^{\frac{1}{k-1}}\right)^{\frac{k-1}{k}} \leq n^{\left(1+\frac{1}{k-1}\right)\frac{k-1}{k}}=n. \qedhere
    \end{equation*}
\end{proof}

We next bound the medium nodes.
\begin{lemma} \label{lem:med}
Let $k\geq 2$ be an integer, let $p = \frac{k}{k-1}$, and let $G = (V, E)$ have girth at least $2k+1$.  Then $\lp{V_{med}} \leq n$.
\end{lemma}
\begin{proof}
From the definition of $V_{med}$, we get that
\begin{align*}
    \|V_{med}\|_p^p &= \sum_{v \in V_{med}} d_1(v)^\frac{k}{k-1} \leq \sum_{v \in V_{med}} n^{\frac{2-k}{k-1}}d_{k-1}(v) d_1(v) = n^{\frac{2-k}{k-1}} \sum_{v \in V_{med}} d_1(v) d_{k-1}(v) \\
    &\leq n^{\frac{2-k}{k-1}}\sum_{v \in V} d_1(v) d_{k-1}(v).
\end{align*}
We claim that $\sum_{v \in V} d_1(v) d_{k-1}(v) \leq n^2$, which would imply the lemma since we would have $\lp{V_{med}} \leq n^{\frac{2-k}{k}} n^{\frac{2k-2}{k}} = n$.  To analyze $\sum_{v \in V} d_1(v) d_{k-1}(v)$, note that
$    \sum_{v \in V} d_1(v) d_{k-1}(v) = \sum_{v \in V} \sum_{w \in N_{k-1}(v)} d_1(v) =  \sum_{v \in V} \sum_{w \in N_{k-1}(v)} d_1(w)$,
where we have used the fact that since the girth is at least $2k+1$, there is exactly one length $k-1$ path between any two nodes at distance $k-1$.  But now, again since the girth is at least $2k+1$, we know that $\sum_{w \in N_{k-1}(v)} d_1(w) = d_k(v) + d_{k-1}(v)$ for all $v \in V$.  Thus
$\sum_{v \in V} d_1(v) d_{k-1}(v) = \sum_{v \in V} (d_k(v) + d_{k-1}(v)) \leq \sum_{v \in V} n = n^2$, as required.
\end{proof}

We now bound the high nodes, with one degree assumption which we will later remove.
\begin{lemma}\label{lem:high}
Let $G=(V,E)$ be a graph of girth at least $2k+1$ with $ k\geq 3$.  Further assume that the graph has minimum degree at least $4$. Then $\| V_{high,j} \|_{k/(k-1)} = O(n)$ for all $0 \leq j \leq \lfloor (k-3)/2\rfloor$.
\end{lemma}
\begin{proof}
We will break the high nodes into the following two sets: 
\begin{align*}
    V'_{high,j} &= \{v \in V_{high,j} : d_{k-2j-1}(v) \geq d_{k-2j-3}(v)d_{1} (v)\}\\
    V''_{high,j} &= \{v \in V_{high,j} : d_{k-2j-1}(v) < d_{k-2j-3}(v)d_{1} (v)\}.
\end{align*}
Obviously $V_{high,j}=V'_{high,j} \cup V''_{high,j}$, so we can bound each of the two sets separately.  For the first set, we get that
\begin{align*}
\|V'_{high,j}\|_{k/(k-1)}
& = \left(\sum_{v \in V'_{high,j}}(d_1(v))^{k/(k-1)} \right)^{(k-1)/k} \leq \left(\sum_{v \in V'_{high,j}}
\frac{n^{\frac{1}{k-1}}d_1(v)^2 d_{k-2j-3}(v)}{d_{k-2j-1}(v)}\right)^{(k-1)/k} \\
&\leq \left(2\sum_{v \in V'_{high,j}}  
\frac{n^{\frac{1}{k-1}}d_1(v)^2 d_{k-2j-3}(v)}{d_{k-2j-1}(v)+d_1(v) d_{k-2j-3}(v)}\right)^{(k-1)/k} \leq 4n.
\end{align*}
The first inequality is from the definition of $V_{high}$, the second is from the definition of $V'_{high}$, and the final inequality is from Corollary~\ref{techratio}.

To analyze $V''_{high}$, note that, by definition, $d_{k-2j-1}(v)+d_1(v)d_{k-2j-3}(v) < 2d_1(v)d_{k-2j-3}(v)$ for all $v \in V''_{high,j}$.
Combining this with Corollary~\ref{techratio} implies that 
\begin{align*}
\|V''_{high,j}\|_{k/(k-1)} &\leq \|V''_{high,j}\|_1 =\sum_{v \in V''_{high,j}}d_1(v) = \sum_{v \in V''_{high,j}}\frac{d_1(v)^2 d_{k-2j-3}(v)}{d_1(v)d_{k-2j-3}(v)} \\
&\leq 2\sum_{v \in V''_{high,j}}\frac{d_1(v)^2 d_{k-2j-3}(v)}{d_{k-2j-1}(v)+d_1(v)d_{k-2j-3}(v)} \leq 4n. 
\end{align*}
Thus $\|V_{high,j}\|_{k/(k-1)} \leq \|V'_{high,j}\|_{k/(k-1)} + \|V''_{high,j}\|_{k/(k-1)} \leq 8n$.
\end{proof}


Putting this all together gives the following theorem.
\begin{theorem} \label{thm:upper-assuming-degree}
Let $G = (V,E)$ have girth at least $2k+1$, $k \geq 2$ and minimum degree at least $4$.  Then $\|G\|_p \leq O(k n)$ for $p =\frac{k}{k-1}$.
\end{theorem}
\begin{proof}
We know from Theorem~\ref{thm:decomposition} that $V=V_{low} \cup V_{med} \cup \left( \cup _{0 \leq j \leq \lfloor (k-3)/2 \rfloor} V_{high,j} \right)$ for $k \geq 3$. Thus $\lp{G} \leq  \lp{V_{low}} + \lp{V_{med}}  +\sum_{j = 0}^{\lfloor (k-3)/2\rfloor} \lp{V_{high,j}} \leq O\left(k n\right)$,
where we used Lemmas~\ref{lem:low}, \ref{lem:med}, \ref{lem:high}, to bound the contribution of each set. If $k=2$ then $V_{med}=V$ and the proof is similar (alternatively see Theorem~\ref{thm:upper-3}).
\end{proof}

We can now remove the degree assumption and the restriction to $p=\frac{k}{k-1}$, to finally prove Theorem~\ref{thm:upper-main}.

\begin{proof}[Proof of Theorem~\ref{thm:upper-main}]
The case of $k=1$ is trivial since every graph $H$ has $\|H\|_p \leq O(n^{\frac{p+1}{p}})$.  For $k \geq 2$, by Lemma~\ref{lem:upper-main-reduction}, we may assume that $p=\frac{k}{k-1}$.  We will use induction on the number of vertices of degree less than $4$.  If $H$ has no vertices with degree less than $4$, then Theorem~\ref{thm:upper-assuming-degree} implies Theorem~\ref{thm:upper-main}.  Otherwise, let $v \in V$ be a vertex of degree at most $3$, and let $G' = G - v$ be the graph obtained by removing $v$.  Then it is easy to see that $\|\vec{d_G} - \vec{d_{G'}}\|_1 \leq 6$, since one entry in the degree vector of value at most $3$ gets removed and at most three other entries get decreased by $1$. Thus we can use triangle inequality and monotonicity of norms to get that $\|G\|_p - \|G'\|_p \leq \|\vec{d_G} - \vec{d_{G'}}\|_p \leq \|\vec{d_G} - \vec{d_{G'}}\|_1 \leq 6$.  Hence by the induction hypothesis we get that $\|G\|_p \leq O\left(k n \right)$ as required.
\end{proof}

\section{Universal Lower Bound} \label{sec:lower-bound-main}
As stated in Theorem~\ref{thm:lower-meta}, our lower bound can be calculated by a simple linear program of size $O(t)$ (where $t$ is the stretch).  We give this linear program formally in Section~\ref{sec:LP}.  The linear program assumes that the graph has a fairly regular structure. In particular, it assumes that the extremal $t$-spanner $H$ is a layered graph with $t+1$ layers $V_0,\ldots, V_{t}$, such that the subgraph induced on every two subsequent layers $V_{i},V_{i+1}$ is bipartite and biregular (in each side, all vertices have the same degree), and that the original extremal graph $G$ (the graph whose spanner $H$ achieves the lower bound) in addition has a biregular graph between $V_0$ and $V_t$ which contributes most of the $p$-norm of $G$, and is spanned by the layered graph $H$. Such a spanner $H$ can be briefly described by the cardinalities of the layers $V_i$ and the degrees of the bipartite graphs connecting every two consecutive layers.

As we show, this assumption is without loss of generality, in the sense that pruning any graph to obtain this structure can change the $p$-norm of the graph or its spanner by at most a polylogarithmic factor. The linear program captures the constraints that the parameters of a spanner with such a regular structure must satisfy. These constraints are also sufficient in the sense that given any solution to the linear program, we can construct a graph $G$ and spanner $H$ of this form with the parameters given by this LP solution.

In fact, the extremal spanners which match our lower bound have a fairly specific structure with consistent properties:
\begin{itemize}
    \item The layers in the extremal can be partitioned into three sections: an initial section in which we have layers of decreasing size $|V_0|\geq |V_1|\ldots\geq|V_{L}|$, a middle section consisting of equal size layers $|V_L|=\ldots=|V_{L+C}|$, and a final section with layers of increasing size $|V_{L+C}|\leq\ldots\leq|V_{L+C+R}|$. In some cases one of the first two sections may be missing.
    \item The bipartite graphs between every two consecutive layers in the spanner have the same contribution to the $p$-norm of the spanner.
    \item In addition to the edges in the spanner, the original graph also contains a biclique between the outer layers $V_0$ and $V_t$, so that $\|G\|_p=\Theta(|V_0|^{1/p}|V_t|)$.
\end{itemize}
The structure of these spanners has the property that given the lengths of the three sections, we can derive the exact structure of the spanner, and hence the exact value of the lower bound. In our analysis, we focus on this specific family of graphs, and show that it suffices to describe our lower bound.

While the lower bound for $p=1$ or $p=\infty$ is simple, it turns out that the lower bound for intermediate values of $p$ is quite complex, and depends on the stretch $t$, the norm parameter $p$, and the $p$-norm of the input graph $\Lambda$ in a highly non-trivial way. To identify the extremal spanners and prove their optimality, we look at the dual of our linear program, and for every graph in our family of candidate extremal spanners, examine whether there exists a dual solution which satisfies complementary slackness w.r.t.\ the primal LP solution corresponding to our spanner. With this approach, for every $p,t,\Lambda$, we are able to identify the exact constraints that the parameters of an optimal spanner from our family must satisfy, and give an explicit solution, which gives our lower bound.

As an example, our analysis identifies the lower bound for relatively low values of $p$:\footnote{The complete description of our lower bound is quite long, but Theorem~\ref{thm:low-p} can be seen to follow from Claim~\ref{clm:LB-defs}, Theorem~\ref{thm:LB-nice}, and the parameters described in Section~\ref{sec:lowest-p}.}

\begin{theorem} If $t$ is even, then for all $p\in[1,\varphi]$ (where $\varphi=\frac{1+\sqrt{5}}{2}$ is the golden ratio),   $$\textstyle\mathrm{LB}^p_t(n,\Lambda)=\tilde\Theta\left(\max\left\{n^{1/p},\Lambda^{\alpha}\right\}\right)\text{ for }\alpha=1/\left((p+1)\left(1-\left((p-1)/p\right)^{t/2}\right)\right).$$
If $t$ is odd, then for all $p\in[1,2]$,  $$\textstyle\mathrm{LB}^p_t(n,\Lambda)=\tilde\Theta\left(\max\left\{n^{1/p},\Lambda^{\beta}\right\}\right)\text{ for }\beta=1/\left(1+p\left(1-\left((p-1)/p\right)^{(t-1)/2}\right)\right).$$
\label{thm:low-p}
\end{theorem}
\begin{corollary} \label{cor:LB-low-stretch}
 For all $p\in[1,\varphi]$, we have $\mathrm{LB}^p_2(n,\Lambda)=\tilde\Theta\left(\max\left\{n^{1/p},\Lambda^{p/(p+1)}\right\}\right).$  For all $p\in[1,2]$, we have
$\mathrm{LB}^p_3(n,\Lambda)=\tilde\Theta\left(\max\left\{n^{1/p},\sqrt{\Lambda}\right\}\right).$
\end{corollary}
Note that the dependence on $n$ for this range of parameters is minimal. In fact, the only dependence on $n$ is due to the fact that any connected $n$-vertex graph (such as the spanner of a connected $n$-vertex graph) must have $p$-norm at least $n^{1/p}$. If we remove the condition that the graph must be connected, the lower bounds in Theorem~\ref{thm:low-p} become $\tilde\Theta(\Lambda^{\alpha})$ and $\tilde\Theta(\Lambda^{\beta})$.

For higher values of $p$, the lower becomes more complex. In particular, the parameters which determine the extremal spanner depend not only on $p$ and $t$, but also on the \emph{$p$-log density} of the graph, which we define to be $\log_n(\Lambda)$. This parameter generalizes the notion of log-density, which is at the heart of several recent breakthroughs in approximation algorithms~\cite{BCCFV10,CDK12,CDM17,CMMV17,CM18}, in which log-density was used to mean $p$-log density for $p=1$ or $p=\infty$. As in that line of work, the structure and parameters of the graphs of interest here (the extremal spanners) is a function of the $p$-log density of our graph which does not depend on $n$. The complete technical details of our lower bound appear in Section~\ref{sec:LB-technical}.

\section{Details of Our Lower Bound}\label{sec:LB-technical}

\subsection{Main Technical Theorems and Overview}
We now focus on analyzing and describing the lower bound $\mathrm{LB}_{t}^p(n, \Lambda)$. As stated in Theorem~\ref{thm:lower-meta}, the main tool in our analysis is a small linear program which calculates this lower bound for any value of $t,p,\Lambda$. For technical reasons, we will focus here on a slightly different, but closely related lower bound:
$$\overline{\mathrm{LB}}_t^p(n,\Lambda)=\min_{\substack{G=(V, E): |V|=n,\\ \|G\|_p = \Lambda}} \;\min_{H:\text{ $H$ is a $t$-spanner of $G$}}\|H\|_p$$

Note the only difference between $\overline{\mathrm{LB}}$ and $\mathrm{LB}$: in the definition of $\overline{\mathrm{LB}}$ we do not require that the graph $G$ be connected, or even that it will not have any isolated vertices. This may seem like a strange choice, since any reasonable analysis of approximation algorithms for spanners (the motivation for our lower bound) would assume wlog that the graph is connected. However, this assumption would make the presentation of our lower bound unwieldy. Fortunately, there is a very simple and straightforward connection between these two definitions:

\begin{claim}
For any $p,t,n$ and $\Lambda\geq 2n^{1/p}$, we have $$\mathrm{LB}_{t}^p(n, \Lambda)=\Theta(\max\{n^{1/p},\overline{\mathrm{LB}}_t^p(n,\Lambda)\}).$$
\label{clm:LB-defs}
\end{claim}
\begin{proof}
  Clearly, this is an lower bound on $\mathrm{LB}_{t}^p(n, \Lambda)$, since if $G$ is connected, any $t$-spanner of $G$ must also be connected, and in particular have minimum degree at least $1$, giving a lower bound of $n^{1/p}$ on the $\ell_p$ norm of the minimum spanner.
  
  On the other hand, this is also an upper bound. Indeed, let $G$ and $H$ be graphs matching the bound in the definition of $\overline{\mathrm{LB}}_t^p(n,\Lambda)\geq n^{1/p}$. If we let $C_1,\ldots,C_s$ be the connected components of $G$, then we can add a path $v_1,\ldots,v_s$ connecting some arbitrary choice of vertices $v_i\in C_i$ to both $G$ and $H$. This will have the effect of making $G$ connected, while adding at most $O(n^{1/p})$ (which is also $O(\Lambda)$) to the $\ell_p$ norm of both $G$ and $H$, which gives the upper bound on $\mathrm{LB}_{t}^p(n, \Lambda)$.
\end{proof}

To understand the structure of our extremal graphs, it will be helpful to consider the $\ell_p$ norm of the graphs and their spanners through the lens of log-density.

\begin{definition} The $p$-log density of an $n$-vertex graph $G$ is defined to be $\log_n\|G\|_p$.
\label{def:plog-density}
\end{definition}

For consistency, we will use $\lambda$ to denote the $\ell_p$ norm of our graph (so $\Lambda=n^{\lambda}$), and $\ell$ to denote the $\ell_p$ norm of the extremal spanner (so $\overline{\mathrm{LB}}_t^p(n,n^{\lambda})=n^{\ell}$).

Note that the $p$-log density of a graph can range up to $1+1/p$ (the $p$-log density of a clique). For most of the possible range of $\lambda$, we have a consistent lower bound:

\begin{theorem} For all $t\in\nats$, $p\geq1$, there exist nonnegative integers $L,C,R (=L(p,t),C(p,t),R(p,t))$, derivable from our LP, such that $L\leq R$ and $L+C+R=t$, and such that for all $\lambda\leq1+E_{C,L}/(pE_{C,R})$, we have $\overline{\mathrm{LB}}_{t}^p(n, n^\lambda)=\tilde\Theta(n^{\ell})$, where
$$\ell=\left\{\begin{array}{ll}\displaystyle\frac{1+p/C}{E_{C,L}+pE_{C,R}}\cdot\lambda&\text{if }C>0,\\
\phantom{.}\\
\displaystyle\frac{p}{E_{1,L}-1+p(E_{1,R}-1)}\cdot\lambda&\text{if }C=0.\end{array}\right.$$ for $$E_{i,j}:=1+\frac{p}{i}\left(1-\left(\frac{p-1}{p}\right)^j\right).$$
\label{thm:LB-nice}
\end{theorem}

The parameters $L,C,R$ can be derived from the linear program, and we will calculate them, as functions of $p,t$, explicitly. Note that this bound does not depend at all on $n$. If $L=R$ above, then it follows that this bound applies to all $\lambda$ up to $1+1/p$. That is, for graphs of every possible $p$-log density. However, as we shall see, this will only be the case for relatively small values of $p$. When $L<R$, there will be a high range of $\lambda$ for which the parameters and graph structures will depend on the $p$-log density $\lambda$, not only on $p,t$.

\begin{theorem} For all $t\in\nats$, $p\geq 1$, let $L,C,R$ be as in Theorem~\ref{thm:LB-nice}, and for all $i=0,1,\ldots,R-L$, define
 $$L_i=L+\lfloor i/2\rfloor\qquad C_i=C+\lceil i/2\rceil-\lfloor i/2\rfloor \qquad R_i=R-\lceil i/2\rceil.$$
 If $L>0$, then for every $i\in[R-L]$, if $$\lambda=1+\theta\cdot\frac{E_{C_{i-1},L_{i-1}}}{pE_{C_{i-1},R_{i-1}}}+(1-\theta)\cdot\frac{E_{C_i,L_i}}{pE_{C_i,R_i}},$$
 then $\mathrm{LB}_{t}^p(n, n^\lambda)=\tilde\Theta(n^{\ell})$, where
 $$\ell=\theta\cdot\left(\frac{1}{p}+\frac{1}{C_{i-1}}\right)\cdot\frac{1}{E_{C_{i-1},R_{i-1}}}+(1-\theta)\cdot\left(\frac{1}{p}+\frac{1}{C_{i}}\right)\cdot\frac{1}{E_{C_{i},R_{i}}}.$$
 If $L=0$, then for every $i\in[R]$, if $$\lambda=1+\theta\cdot\frac{1}{pE_{C+i-1,R-i+1}}+(1-\theta)\cdot\frac{1}{pE_{C+i,R-i}},$$
\label{thm:LB-not-nice}
\end{theorem}

\subsection{Formal linear program}\label{sec:LP}

Up to polylogarithmic factors, the following linear program expresses the minimum possible $\ell_p$ norm of a $t$-spanner of an $n$-vertex graph with $p$-log-density at least $\lambda$ (i.e., $\ell_p$ norm at least $n^\lambda$):
\\

\noindent {\bf LP$(n,n^{\lambda},p,t)$:}
\begin{align} \min\qquad&n^\ell\nonumber\\
\text{s.t.}\qquad
&n_{i-1}^{1/p}d_i\leq n^\ell&\forall i\in[t]\label{LP1:left-norm}\\
&n_i^{1/p}(n_{i-1}d_i/n_i) \leq n^\ell&\forall i\in[t]\label{LP1:right-norm}\\
&d_i \leq n_i&\forall i\in[t]\label{LP1:degree-size}\\
&n_{i-1}d_i\geq n_i&\forall i\in[t]\label{LP1:right-deg}\\
&\Delta_1=d_1\label{LP1:d_1}\\
&\Delta_i\leq \Delta_{i-1}d_i &\forall i\in\{2,\ldots,t\}\label{LP1:i-degree}\\
&\Delta_i\leq n_i &\forall i\in\{2,\ldots,t\}\label{LP1:i-degree-size}\\
&n_0^{1/p}\Delta_t\geq n^\lambda\label{LP1:G-norm}\\
&n_i\leq n&\forall i\in\{0,\ldots,t\}\label{LP1:size-bound}\\
&n^\lambda,n_0,\ldots,n_t,d_1,\ldots,d_t,\Delta_1,\ldots,\Delta_t\geq 1
\end{align}

While the constraints are not linear, if we think of $\lambda$ as a constant and take $\log$ base $n$ of all the expressions, this becomes a linear program in the variables $\ell,\log_n n_0,\ldots,\log_n n_t,\log_n d_1,\ldots,\log_n d_t$, $\log_n\Delta_1,\ldots,\log_n\Delta_t$. Note that after this manipulation, we get a linear program which is independent of $n$. It is only a function of $p$, $t$, and the $p$-log density, $\lambda$. Note also that this is indeed a linear program of size $O(t)$. To prove Theorem~\ref{thm:lower-meta}, we need to show that it does in fact compute $\overline{\mathrm{LB}}_t^p(n,\Lambda)$ (which gives us $\mathrm{LB}_t^p(n,\Lambda)$ by Claim~\ref{clm:LB-defs}). The proof of the theorem follows from two basic lemmas. First, we show that the LP gives an upper bound on $\overline{\mathrm{LB}}$:

\begin{lemma} For any $p,t,n,\Lambda$, any feasible solution to $\mathrm{LP}(n,n^{\lambda},p,t)$ corresponds to a $t$-spanner with $\ell_p$ norm $O(n^{\ell}\log n)$ of an $n$-vertex graph with $\ell_p$ norm at least $n^{\lambda}$.
\end{lemma}
\begin{remark}
The additional $O(\log n)$ factor is not necessary, since at the very least for our family of optimal solutions, we can easily construct matching graphs with at most a constant loss (for fixed $t$). However, we do not focus on this point since at any rate there is a polylogarithmic loss in the other direction.
\end{remark}
\begin{proof}
 Let $\lambda,n_0,\ldots,n_t,d_1,\ldots, d_t,\Delta_1,\ldots,\Delta_t$ be a feasible solution. Then we construct a spanner $H$ as follows: Define disjoint vertex layers $V_0,\ldots,V_t$ of size $n_0,\ldots,n_t$, respectively. If these layers contribute less than $n$ vertices, add isolated vertices to reach size $n$. If the contribute more, then note by constraint~\eqref{LP1:size-bound} that they contribute at most  $tn=O(n)$, so the after the construction we can scale the whole graph down by at most a constant factor $t$. Now for every $i\in[t]$, for every vertex $v\in V_{i-1}$, independently add edges from $v$ to $\min\{d_i\log n,n_i\}$ random neighbors in $V_i$. The graph $G$ is then defined as follows: $G$ includes all the edges in $H$, and in addition, for all $u\in V_0$, $v\in V_t$, add an edge $(u,v)$ iff there is a path of length $t$ from $u$ to $v$ in $H$. Note that $H$ is clearly a $t$-spanner of $G$ by construction.
 
 For every $i\in[t]$, every vertex in $V_{i-1}$ has $O(d_i\log n)$ neighbors in $V_i$ by definition, and w.h.p.\ every vertex in $V_i$ has $O((n_{i-1}d_i/n_i)\log n)$ neighbors in $V_{i-1}$ by a simple Chernoff bound and constraint~\eqref{LP1:right-deg}. Thus by constraints~\eqref{LP1:left-norm} and~\eqref{LP1:right-norm}, the $\ell_p$ norm of $H$ is bounded by $O(t^{1/p}n^{\ell}\log n)$. Since $t$ is fixed, it remains only to show that the $\ell_p$ norm of $G$ is at least $n^{\lambda}$.
 
 To show this, we will show that from every $v\in V_0$, for every $i\in[t]$, there are paths of length $i$ to at least $\Delta_i$ nodes in $V_i$. In particular, this means that every vertex $v\in V_0$ has degree at least $\Delta_t$ in $G$, which by constraint~\eqref{LP1:G-norm} gives the desired lower bound on $\|G\|_p$. The claim can be shown by induction. For $i=1$, it follows immediately from the definition of $H$ and constraints~\eqref{LP1:degree-size} and~\eqref{LP1:d_1} that the degree of $v$ is at least $d_1$. Now let $i\in[t-1]$, and let $S_i$ be the set of vertices in $V_i$ reachable from $v$ by a path of length $i$. By the inductive hypothesis, we have $|S_i|\geq\Delta_i$. Since by construction, every vertex $u\in\S_i$ has $\min\{d_{i+1}\log n,n_{i+1}\}$ random neighbors in $V_{i+1}$, it follows by a simple Chernoff bound that $S_i$ has at least $\max\{n_{i+1},d_{i+1}|S_i|\}$ neighbors in $V_{i+1}$. By constraints~\eqref{LP1:i-degree} and~\eqref{LP1:i-degree-size}, this is at least $\Delta_{i+1}$. Thus $v$ has paths of length $i+1$ to at least $\Delta_{i+1}$ nodes in $V_{i+1}$.
\end{proof}

Next, we show that the LP gives an lower bound on $\overline{\mathrm{LB}}$:

\begin{lemma} Let $G$ by an $n$-vertex graph with $\|G\|=n^{\lambda}$, and $H$ a $t$-spanner of $G$. Then $\mathrm{LP}(n,n^{\lambda},p,t)\leq\|H\|_p\log^{O(t)}(n)$.
\end{lemma}
\begin{proof}
 Let $G$ and $H$ be as above. We will use a common bucketing and pruning argument to transform $H$ into a layered graph composed of a sequence of nearly regular graphs, which approximately satisfies the LP constraints. By our assumption, $\sum_{v\in V}d_G(v)^p=n^{\lambda p}$. If we partition the non-isolated vertices of $G$ by their degrees by defining $B_j=\{v\in V\mid 2^{j-1}\leq\deg(v)<2^j$ for all $j\in[\lceil\log n\rceil]$, then there is some choice of $\hat j$ such that $\sum_{v\in B_{\hat j}}d_G(v)=\Omega(n^{\lambda p}/\log n)$, and all vertices in $B_{\hat j}$ have roughly the same degree $\Delta:=2^{\hat j-1}$ (up to a factor 2). In particular, $|B_{\hat j}|\Delta^p=\Omega(n^{\lambda p}/(2^p\log n))$.
 
 We will construct a subgraph $\hat H$ of $H$ as follows: For every vertex $u\in B_{\hat j}$, and every $v\in N_G(u)$, $H$ contains a path of length at most $t$ from $u$ to $v$. W.l.o.g.\ we will assume this path is of length exactly $t$, by duplicating vertices. Since $t$ is fixed, this will affect the norms of our graphs by a constant factor, which we will ignore. Start with $\hat H$ as the union of all these paths for all $u\in B_{\hat j}$ and all $v\in N_G(u)$. This is a layered graph with $t+1$ layers $\hat V_0,\ldots, \hat V_t$ where for now $\hat V_0=B_{\hat j}$.
 
 Now perform the following pruning procedure:
 \begin{itemize}
     \item For all $i=t-1$ down to 0
     \begin{itemize}
         \item For every $u\in \hat V_i$, define $$\hat N^i_u(v):=\{w\in N_G(v)\mid u\text{ is the $i$th vertex on a path of length $t$ from $v$ to $w$ in }\hat H\}$$
         (note that these paths must be in the current version of $\hat H$).
         \item Partition $\hat V_i$ into $B^i_j=\{u\in \hat V_j\mid 2^{j-1}\leq |N_{\hat H}(u)\cap \hat V_{i+1}|<2^j\}$, for $j\in [\lceil\log n\rceil]$.
         \item Let $j_{i}$ be the value of $j$ that maximizes the quantity $\sum_{v\in \hat V_0}\left|\bigcup_{u\in B^i_j}\hat N^i_u(v)\right|^p$.
         \item Delete from $\hat V_i$ all vertices not in $B^i_{j_i}$ and their incident edges.
     \end{itemize}
 \end{itemize}

Note that after this pruning procedure, for all $i\in[t]$, the number of neighbors in $\hat V_i$ of every vertex in $\hat V_{i-1}$ is between $\hat d_i$ and $\hat 2d_i$ for $\hat d_i=2^{j_{i-1}-1}$. Let us examine what happens to the quantity 
\begin{equation}
  \sum_{v\in \hat V_0}|N^t_{\hat H}(v)\cap N_G(v)|^p    
  \label{eq:pruning-norm}
\end{equation}
 after every iteration of the pruning loop. Note that we have
\begin{align*}
    \sum_{v\in \hat V_0}|N^t_{\hat H}(v)\cap N_G(v)|^p&=\sum_{v\in \hat V_0}\left|\bigcup_{u\in\hat V_i}\hat N^i_u(v)\right|^p\\
    &= \sum_{v\in \hat V_0}\left|\bigcup_{j=1}^{\lceil\log n\rceil}\bigcup_{u\in B^i_j}\hat N^i_u(v)\right|^p\\
    &\leq \sum_{v\in \hat V_0}\left(\sum_{j=1}^{\lceil\log n\rceil}\left|\bigcup_{u\in B^i_j}\hat N^i_u(v)\right|\right)^p\\
    &\leq \sum_{v\in \hat V_0}\lceil\log n\rceil^{p-1}\sum_{j=1}^{\lceil\log n\rceil}\left|\bigcup_{u\in B^i_j}\hat N^i_u(v)\right|^p&\text{by convexity of }f(x)=x^p\\
    &=\lceil\log n\rceil^p\cdot \frac{1}{\lceil\log n\rceil}\sum_{j=1}^{\lceil\log n\rceil}\sum_{v\in \hat V_0}\left|\bigcup_{u\in B^i_j}\hat N^i_u(v)\right|^p.
\end{align*}
Thus, there is at least one choice of $j\in\lceil\log n\rceil$ such that deleting from $\hat V_i$ all vertices not in $B^i_{j}$ reduces~\eqref{eq:pruning-norm} by at most a factor of $\lceil\log n\rceil^p$. After the loop is completed, the quantity~\eqref{eq:pruning-norm} has been reduced by at most $\lceil\log n\rceil^{pt}$.

We can now define the following LP solution, based on the pruned graph $\hat H$: Let $\hat d_i$ be as above, let $\hat n_i=|\hat V_i\|$, and inductively define $\hat\Delta_1=\hat d_1$, and $\hat\Delta_i=\min\{\hat n_i,\hat\Delta_{i-1}\hat d_i\}$. Let us consider the various LP constraints. For constraint~\eqref{LP1:left-norm}, note that $\hat n_{i-1}\hat d_i\leq \left(\sum_{v\in \hat V_{i-1}}|N_{\hat H}(v)\cap \hat V_{i}|^p\right)^{1/p}\leq\|H\|_p$. Constraint~\eqref{LP1:right-norm} is less immediate, since the nodes in $\hat V_{i}$ don't all have roughly the same number of neighbors in $\hat V_{i-1}$. However, on average, the have at least $\hat n_{i-1}\hat d_{i-1}/\hat n_i$ neighbors in $\hat V_{i-1}$, and their contribution to the $\ell_p$ norm is minimized when they all have the same degree, so this bounds the left hand side again by $\|H\|_p$. Constraint~\eqref{LP1:degree-size} follows since clearly no node can have more neighbors in $\hat V_i$ than the number of vertices in $\hat V_i$. Constraints~\eqref{LP1:d_1} through~\eqref{LP1:i-degree-size} follow by construction, and constraint~\eqref{LP1:size-bound} follows since $\hat V_i$ is a subset of $V$. Finally, consider constraint~\eqref{LP1:G-norm}. It is easy to see by induction that every vertex in $\hat V_0$ can reach at most $2^i\hat\Delta_i$ nodes in $\hat V_i$ by a paths of length $i$. This follows since $\hat V_i$ only contains $\hat n_i$ nodes, and since if it reaches $2^{i-1}\hat\Delta_{i-1}$ nodes in $V_{i-1}$, each of those can account for at most $2\hat d_i$ nodes it can reach in $\hat V_i$. Thus, it can reach at most $2^t\hat\Delta_t$ nodes in $\hat V_t$, and so the quantity~\eqref{eq:pruning-norm} is at most $\hat n_0 2^{tp}\hat\Delta^{tp}$. On the other hand, we've also shown that this quantity is at least $\Omega(n^{\lambda p}/(2^p(\log n)^{pt+1})$. Thus we have $\hat n_0^{1/p}\hat\Delta_t=\Omega(n^\lambda/(2^t(\log n)^{t+1/p}))$.

Thus, this solution satisfies all the constraints except for constraint~\eqref{LP1:G-norm}, which is violated by at most a $\log^{O(t)}(n)$ factor. Thus, letting $\sigma:=n^{\lambda}/(\hat n_0^{1/p}\hat\Delta_t)=\log^{O(t)}(n)$, we can define a new solution $d_i:=\min\{\sigma\hat d_i,n\}$, $n_i:=\min\{\sigma\hat n_i,n\}$, and $\Delta_i:=\min\{\sigma\hat\Delta_i,n\}$, and it is not hard to see that this is a feasible LP solution. Since it increases the left hand side of constraints~\eqref{LP1:left-norm} and~\eqref{LP1:right-norm} by at most a $\sigma^{1+1/p}=\log^{O(t)}(n)$ factor, this gives the desired bound on the LP value.
\end{proof}


We relax this LP by eliminating a number of constraints, and combining other constraints into a new LP. As we will see, the new, simpler LP always has an optimal solution which satisfies all the above constraints, so the simpler LP gives the same bound. In particular, we eliminate constraint~\eqref{LP1:degree-size}, 
and constraints~\eqref{LP1:i-degree-size} and~\eqref{LP1:size-bound} for all $i$ except for $i=t$. We also combine constraints~\eqref{LP1:d_1} and~\eqref{LP1:i-degree} into a single constraint by multiplying them out (getting $\prod_{i=1}^t d_i\geq\Delta_t$), and write $\Delta=\Delta_t$. With these modifications, and rewriting the remaining constraints in standard form, we get the following LP:

\begin{align} \min\qquad&n^\ell\\
\text{s.t.}\qquad&\Delta^{-1}\prod_{i=1}^t  d_i \geq 1\label{LP2:spanning}\\
&n^\ell \cdot n_{i-1}^{-1/p}d_i^{-1}\geq 1&\forall i\in[t]\label{LP2:left-deg-norm}\\
&n^\ell \cdot n_i^{(p-1)/p}n_{i-1}^{-1}d_i^{-1} \geq 1&\forall i\in[t]\label{LP2:right-deg-norm}\\
&n_i^{-1}n_{i-1}d_i\geq 1&\forall i\in[t]\label{LP2:right-deg-positive}\\
&n_0^{1/p}\Delta\geq n^\lambda\label{LP2:Lambda}\\
&n_t\Delta^{-1}\geq 1\label{LP2:final-layer}\\
&n_t^{-1}\geq n^{-1}\label{LP2:n}\\
&n^\ell,\Delta,n_0,\ldots,n_t,d_1,\ldots,d_t\geq 1
\end{align}

If we associate dual variables $x,a_1,\ldots,a_t,b_1,\ldots,b_t,D_1,\ldots,D_t,y,w,s$ with the above primal constraints (in that order), we get the following dual:
\begin{align}
\max\qquad&n^{\lambda y-s}\nonumber\\
\text{s.t.}\qquad&\sum_{i=1}^t(a_i+b_i)\leq 1&\text{corresponding to $\ell$}\label{LP3:lambda}\\
&y\leq x+w&\text{corresp.\ to $\Delta$}\label{LP3:Delta}\\
&p^{-1}y+D_1\leq p^{-1}a_1+b_1&\text{corresp.\ to $n_0$}\label{LP3:n_0}\\
&\frac{p-1}{p}\cdot b_i+D_{i+1}\leq p^{-1}a_{i+1}+b_{i+1}+D_i&\text{corresp.\ to $n_i,1\leq i\leq t-1$}\label{LP3:n_i}\\
&\frac{p-1}{p}\cdot b_t+w\leq D_t+s&\text{corresp.\ to $n_t$}\label{LP3:n_t}\\
&x+D_i\leq a_i+b_i&\text{corresp.\ to $d_i,i\in[t]$}\label{LP3:d_i}\\
&x,a_1,\ldots,a_t,b_1,\ldots,b_t,D_1,\ldots,D_t,y,w,s\geq 0
\end{align}

For every stretch $t$, value of $p$, and $p$-log-density $\lambda$, we give an explicit optimal solution to the primal LP (corresponding to an optimal graph, up to polylogarithmic factors). By complementary slackness, it suffices to find a dual solution such that for every non-tight primal constraint, the corresponding dual variable is $0$, and for every primal variable strictly greater than 1, the corresponding dual constraint is tight.

\subsection{Optimal Solutions and Corresponding Duals}

We consider first the following family of solutions:
\begin{definition}
An $(L,C,R)$ minimal spanner is a layered graph with $L+C+R+1$ layers of size $n_0\geq n_1\geq\ldots\geq n_L=n_{L+1}=\ldots=n_{L+C}\leq n_{L+C+1}\leq\ldots\leq n_{L+C+R}$. In every layer $i$, $i\in[L]$, every vertex has exactly $n_{i-1}/n_i$ unique neighbors in layer $i-1$ (the induced subgraph is a collection of disjoint stars, and $d_1=\ldots=d_L=1$), in every layer $L+C+i$, $i\in\{0,\ldots,R-1\}$, every vertex has exactly $d_{L+C+i+1}=n_{L+C+i+1}/n_{L+C+i}$ unique neighbors in layer $i+1$, and for every $i\in[C]$, the induced subgraph on layers $L+i-1,L+i$ is regular with degree $d_{L+1}=\ldots=d_{L+C}=n_L^{1/C}$, so that from every vertex in layer $L$ there is a unique path of length $L$ to every vertex in layer $L+C$. Finally, the layer sizes are set so that the subgraph induced on every two consecutive layers has roughly the same contribution to the $\ell_p$ norm. That is, $n_1(n_0/n_1)^p=\ldots=n_L(n_{L-1}/n_L)^p=n_L d_{L+C+1}^p=n_{L+C}d_{L+C+1}
^p=\ldots=n_{L+C+R-1}d_{L+C+R}^p$.
\end{definition}

Note that an $(L,C,R)$ minimal spanner has $\ell_p$ norm $\Theta(n_L^{1/p+1/C})$ (unless $C=0$, in which case we set $n_L=n_{L+C}=1$, and then the $\ell_p$ norm is $n_{L-1}=n_{L+1}$). It is also an $(L+C+R)$-spanner of the graph obtained by adding to the current graph edges between every vertex in layer 0 and every vertex in layer $L+C+R$. This graph has $\ell_p$ norm $\Theta(n_0^{1/p}n_{L+C+R})$. The proof of the following lemma shows that the condition $n_0^{1/p}n_{L+C+R}=n^{\lambda}$ determines the values of all other parameters in an $(L,C,R)$ minimal spanner, including the size of the largest layer, $n_{L+C+R}$. As we will see, for every value $p\geq 1$ and every stretch $t$, there is some setting of $L,C,R$ so that $L+C+R=t$, and an $(L,C,R)$-minimal spanner give an optimum solution to the LP, for all $\lambda$ up to the value for which $n_{L+C+R}=n$. Thus, once we show this, Theorem~\ref{thm:LB-nice} will follow from the following lemma.

\begin{lemma} For any non-negative integers $L,C,R$, and $\lambda>0$, an $(L,C,R)$-minimal spanner of a graph with $\ell_p$ norm $n_0^{1/p}n_{L+C+R}=n^{\lambda}$ has $\ell_p$ norm $\Theta(n^\ell)$, where
$$\ell=\left\{\begin{array}{ll}\displaystyle\frac{1+p/C}{E_{C,L}+pE_{C,R}}\cdot\lambda&\text{if }C>0,\\
\phantom{.}\\
\displaystyle\frac{p}{E_{1,L}-1+p(E_{1,R}-1)}\cdot\lambda&\text{if }C=0.\end{array}\right.$$ for $$E_{i,j}:=1+\frac{p}{i}\left(1-\left(\frac{p-1}{p}\right)^j\right).$$
\label{lem:LCR-LB}
\end{lemma}
\begin{proof}
  First, note that if we have a sequence of vertex layers of increasing size $U_0,\ldots,U_j$ where for all $i$, every node in $U_i$ has $\tilde d_{i+1}=|U_{i+1}|/|U_i|$ neighbors in $U_{i+1}$, \emph{and} for all $i$ we have the same contribution $|U_i|\tilde d_{i+1}^p$ to the $\ell_p$ norm, then for all $i$ we have
  $$|U_i|(|U_{i+1}|/|U_i|)^p=|U_{i+1}|(|U_{i+2}|/|U_{i+1}|)^p\quad\Rightarrow\quad \tilde d_{i+2}=|U_i|^{1/p}(|U_{i+1}|/|U_i|)/|U_{i+1}|^{1/p}=\tilde d_{i+1}^{(p-1)/p}.$$
  In particular, in such a sequence we have $$|U_j|=|U_0|\prod_{i=1}^j\tilde d_i=|U_0|\tilde d_0^{\sum_{i=0}^{j-1}((p-1)/p)^{i}}=|U_0|\tilde d_0^{E_{1,j}-1}.$$
  Thus, if $C>0$, since nodes in layer $L+C$ have $n_L^{1/C}$ neighbors in layer $L+C+1$, and nodes in layer $L$ have $n_L^{1/C}$ neighbors in layer $L-1$, we have $$n_0=n_L\cdot(n_L^{1/C})^{E_{1,L}-1}=n_L^{E_{C,L}}\qquad\text{and}\qquad n_R=n_L\cdot(n_L^{1/C})^{E_{1,R}-1}=n_L^{E_{C,R}}.$$
  Thus, $n^{\lambda}=n_L^{E_{C,L}/p+E_{C,R}}$. Note that the spanner has $\ell_p$ norm $\Theta(n_L^{1+1/C})$, where $n_L^{(1+1/C)}$ we can rewrite as
  $$\left(n_L^{E_{C,L}/p+E_{C,R}}\right)^{(1+1/C)/(E_{C,L}/p+E_{C,R})}=n^{\lambda\cdot((1+1/C)/(E_{C,L}/p+E_{C,R}))},$$ which gives the result for $C>0$.
  
  For $C=0$, the norm of the spanner is $\Theta(n^{L-1})=\Theta(n^{L+1})=d_{L+1}$. To get the same contribution to the $\ell_p$ norm in every layer, we must have
  $$n_{L+1}d_{L+2}^p=d_{L+1}^p\quad\Rightarrow\quad d_{L+2}=d_{L+1}/n_{L+1}^{1/p}=d_{L+1}^{(p-1)/p},$$ and similarly any node in layer $L-1$ has $d^{L+2}=d_{L+1}^{(p-1)/p}$ neighbors in layer $L-2$. Thus, by the same analysis as before, we have
  $$n_{L+R}=n_{L+1}d_{L+2}^{E_{1,R-1}-1}=d_{L+1}^{1+((p-1)/p)\cdot(E_{1,R-1}-1)}=d_{L+1}^{E_{1,R}-1},$$ and similarly
  $n_0=d_{L+1}^{E_{1,L}-1}$. Thus, if $n^{\lambda}=n_0^{1/p}n_{L+R}=d_{L+1}^{(E_{1,L-1}-1)/p+E_{1,R-1}-1}$, then the $\ell_p$ norm of the spanner is $\Theta(n_{L+1})$, where
  $$n_{L+1}=(n^{\lambda})^{1/((E_{1,L-1}-1)/p+E_{1,R-1}-1)},$$ which completes the proof.
\end{proof}

\subsection{The general solution for low $p$, $\Lambda$}
Here we examine for which parameter settings an $(L,C,R)$ minimal spanner is an optimal solution, when $L,C,R>0$. Note that in an $(L,C,R)$ minimal spanner (assuming $n_t<n$) the only tight constraints are~\eqref{LP2:spanning},~\eqref{LP2:left-deg-norm} for $L+1\leq i\leq L+C+R$,~\eqref{LP2:right-deg-norm} for $i\in[L+C]$,~\eqref{LP2:right-deg-positive} for $L+C+1\leq i\leq L+C+R$,~\eqref{LP2:Lambda}, and~\eqref{LP2:final-layer}. Thus, by complementary slackness, the only (possibly) non-zero dual variables in an optimal dual solution are $x,a_{L+1},\ldots,a_{L+C+R},b_1,\ldots,b_{L+C},D_{L+C+1},\ldots,D_{L+C+R},y$ and $w$. Also, all the primal variables are greater than $1$ except for $d_1,\ldots,d_L$, so by complementary slackness, in an optimal dual solution, all the constraints corresponding to other primal variables must hold with equality. To summarize, an $(L,C,R)$ minimal spanner is optimal iff there exists a non-negative solution to the following system of equations and inequalities:

\begin{align}
&\sum_{i=L+1}^{L+C+R}a_i+\sum_{i=1}^{L+C}b_i= 1&\label{LP:general-scaling}\\
&y= x+w&\\
&p^{-1}y= b_1&\\
&\frac{p-1}{p}\cdot b_i= b_{i+1}&\forall i\in[L-1]\\
&\frac{p-1}{p}\cdot b_i= p^{-1}a_{i+1}+b_{i+1}&\forall i\in\{L,\ldots, L+C-1\}\\
&\frac{p-1}{p}\cdot b_{L+C}+D_{L+C+1}= p^{-1}a_{L+C+1}\\
& D_{i+1}= p^{-1}a_{i+1}+D_i&\forall i\in\{L+C+1,\ldots,L+C+R-1\}\\
&w= D_{L+C+R}\\
&x\leq b_i&\forall i\in[L]\label{LP:general-ineq}\\
&x= a_i+b_i&\forall i\in\{L+1,\ldots,L+C\}\\
&x+D_i= a_i&\forall i\in\{L+C+1,\ldots,L+C+R\}
\end{align}

It is not hard to check that there is a unique (not necessarily non-negative) solution to the above constraints excluding~\eqref{LP:general-ineq}. For some $\eps>0$, this solution is
\begin{align*}
x&=\left(1+p\left((p-1)/p\right)^{R-L}\right)\eps\\
a_i&=\left(1+\frac{i-L}{p-1}\right)x - \frac{C+p}{p-1}\cdot\eps&\forall i\in\{L+1,\ldots,L+C\}\\
a_i&=(C+p)\left((p-1)/p\right)^{R+C-i}\cdot\eps&\forall i\in\{L+C+1,\ldots,L+C+R\}\\
b_i&=\frac{C+p}{p-1}\cdot(p/(p-1))^{L-i}\cdot\eps&\forall i\in[L]\\
b_i&=\frac{C+p}{p-1}\cdot\eps-\frac{i-L}{p-1}\cdot x&\forall i\in\{L,\ldots,L+C\}\\
D_i&=(C+p)\left((p-1)/p\right)^{R+C-i}\cdot\eps-x&\forall i\in\{L+C+1,\ldots,L+C+R\}\\
y&=(C+p)\left(p/(p-1)\right)^L\cdot\eps\\
w&=(C+p)\left(p/(p-1)\right)^L\cdot\eps-x
\end{align*}
If constraint~\eqref{LP:general-ineq} holds, and all these values are non-negative (assuming $\eps>0$), then we can get a dual solution satisfying complementary slackness by choosing $\eps$ which satisfies constraint~\eqref{LP:general-scaling}:
$$\eps=\left(C\left(1+p\left(\frac{p}{p-1}\right)^{R-L}\right)+(C+p)\cdot\left(\left(\frac{p}{p-1}\right)^L-1+\frac{p^{L-1}}{(p-1)^{L-2}}-\frac{(p-1)^{R-L+2}}{p^{R-L+1}}\right)\right)^{-1}$$
Let us first check the condition for the above solution satisfying constraints~\eqref{LP:general-ineq}. Since in the above solution, the sequence $\{b_1,\ldots,b_L\}$ is monotonically decreasing, this is equivalent to the condition $x\leq b_L$, or $$1+p\left(\frac{p-1}{p}\right)^{R-L}\leq\frac{C+p}{p-1},$$
or equivalently
\begin{equation}
    (C+1)\left(\frac{p}{p-1}\right)^{R-L+1}\geq p^2
    \label{eq:LCR-cond1}
\end{equation}

Now consider the conditions for non-negativity in the above solution. $x$ is non-negative by definition. The sequence $a_{L+1},\ldots,a_{L+C}$ is monotonically increasing. Thus for these values it suffices to check that $a_{L+1}\geq 0$, or
$$\frac{p}{p-1}\cdot\left(1+p\cdot\left(\frac{p-1}{p}\right)^{R-L}\right)-\frac{C+p}{p-1}\geq 0,$$
or equivalently
\begin{equation}
    C\cdot\left(\frac{p}{p-1}\right)^{R-L}\leq p^2
    \label{eq:LCR-cond2}
\end{equation}

The values $a_{L+C+1},\ldots,a_{L+C+R}$ are again non-negative by definition, as are $b_1,\ldots,b_L$. The sequence $b_L,\ldots,b_{L+C}$ is monotonically decreasing, so for these values it suffices to check that $b_{L+C}\geq 0$, or
$$\frac{C+p}{p-1}-\frac{C}{p-1}\cdot\left(1+p\cdot\left(\frac{p-1}{p}\right)^{R-L}\right)\geq 0,$$
or equivalently
\begin{equation}
    \left(\frac{p}{p-1}\right)^{R-L}\geq C
    \label{eq:LCR-cond3}
\end{equation}

Note that $y$ is non-negative by definition, and $w=D_{L+C+R}$. This leaves the sequence of values $D_{L+C+1},\ldots,D_{L+C+R}$, which is monotonically increasing, so it remains to check when $D_{L+C+1}\geq 0$. This occurs when
$$(C+p)\left(\frac{p-1}{p}\right)^{R-L-1}-\left(1+p\cdot\left(\frac{p-1}{p}\right)^{R-L}\right)\geq 0,$$
or equivalently
\begin{equation}
    \left(\frac{p}{p-1}\right)^{R-L-1}\leq C+1
    \label{eq:LCR-cond4}
\end{equation}

Finally, note that conditions~\eqref{eq:LCR-cond1} and~\eqref{eq:LCR-cond4} together imply $C\geq p-2$, and conditions~\eqref{eq:LCR-cond2} and~\eqref{eq:LCR-cond3} together imply $C\leq p$. So we have
\begin{equation}
    p-2 \leq C \leq p
    \label{eq:LCR-cond5}
\end{equation}

For $p$ in the range where $L,C,R>0$ above, we can give an explicit solution for $L,C,R$ that satisfies conditions \eqref{eq:LCR-cond1}-\eqref{eq:LCR-cond5}. Start by defining 
$$\Delta_0=
\frac{\log(p^2/\lfloor p\rfloor )}{\log(p/(p-1))}
\quad\Delta_1=
\frac{\log(p\lfloor p\rfloor/(p-1))}{\log(p/(p-1))}
\quad \Delta^+=\max\{\Delta_0,\Delta_1\}\quad\Delta^-=\min\{\Delta_0,\Delta_1\}$$
If $\lfloor\Delta^+\rfloor>\lfloor\Delta^-\rfloor$ (it is easy to check that $\Delta^+-\Delta^-\leq 1$), then in our solution we let
$$C=\lfloor\sqrt{p(p-1)}\rfloor$$
$$L=\lfloor(t-C-\lfloor\Delta^-\rfloor)/2\rfloor$$
$$R=\lceil(t-C+\lfloor\Delta^-\rfloor)/2\rceil.$$
Note that $R-L\in\{\lfloor\Delta^-\rfloor,\lfloor\Delta^+\rfloor\}$.

Otherwise, if $\lfloor\Delta^-\rfloor=\lfloor\Delta^+\rfloor$, we let
$$L=\lceil(t-p-\lfloor\Delta^-\rfloor)/2\rceil$$
$$R=\lceil(t-p+\lfloor\Delta^-\rfloor)/2\rceil$$
$$C=t-L-R$$
Note that here $R-L=\lfloor\Delta^-\rfloor=\lfloor\Delta^+\rfloor$ and that $C\in\{\lfloor p\rfloor,\lfloor p-1\rfloor\}$.

Let us see that these solutions all satisfy conditions \eqref{eq:LCR-cond1}-\eqref{eq:LCR-cond5}. Clearly, condition \eqref{eq:LCR-cond5} is satisfied by definition. 
The following is easy to check.

\begin{observation}\label{obs:simplified-LB} Conditions \eqref{eq:LCR-cond2}, \eqref{eq:LCR-cond4} are satisfied when $C=\lfloor p\rfloor$ and $R-L=\lfloor\Delta_0\rfloor$ or when $C=\lfloor p-1\rfloor$ and $R-L=\lfloor\Delta_1\rfloor$. Conditions \eqref{eq:LCR-cond1} and \eqref{eq:LCR-cond3} are satisfied when $C=\lfloor p\rfloor$ and $R-L=\lfloor\Delta_1\rfloor$ or when $C=\lfloor p-1\rfloor$ and $R-L=\lfloor\Delta_0\rfloor$.
\end{observation}

Now consider the case where $\lfloor\Delta^+\rfloor>\lfloor\Delta^-\rfloor$. We separate this into two cases:

\paragraph{Case 1:} $\lfloor p\rfloor\leq\sqrt{p(p-1)}$. In this case, we get $$C=\lfloor p\rfloor\qquad \Delta^+=\Delta_0\qquad\Delta^-=\Delta_1<\Delta^+$$

\paragraph{Case 2:} $\lfloor p\rfloor>\sqrt{p(p-1)}$. In this case, we get $$C=\lfloor p-1\rfloor\qquad \Delta^+=\Delta_1\qquad\Delta^-=\Delta_0<\Delta^+$$

In both cases, Observation~\ref{obs:simplified-LB} implies that the solutions in which we set $R-L=\lfloor\Delta^+\rfloor$ satisfy conditions \eqref{eq:LCR-cond2}, \eqref{eq:LCR-cond4}, and the solutions in which we set $R-L=\lfloor\Delta^-\rfloor$ satisfy conditions \eqref{eq:LCR-cond1}, \eqref{eq:LCR-cond3}. The other conditions are satisfied in these solutions by monotonicity, since $\Delta^-<\Delta^+$.

Now consider the case where $\lfloor\Delta^-\rfloor=\lfloor\Delta^+\rfloor$. In this case, we have $\lfloor\Delta^-\rfloor=\lfloor\Delta^+\rfloor=\lfloor\Delta_0\rfloor=\lfloor\Delta_1\rfloor$. Then in both solutions (where $C\in\{\lfloor p-1\rfloor,\lfloor p\rfloor\}$), conditions \eqref{eq:LCR-cond1}-\eqref{eq:LCR-cond4} are satisfied by Observation~\ref{obs:simplified-LB} and the fact that $\lfloor\Delta_0\rfloor=\lfloor\Delta_1\rfloor$.

Finally, we will need the following observation.

\begin{observation}\label{obs:LCR-continuity} $\Delta^-,\Delta^+$ are continuous functions of $p$. This, combined with our choice of $L,C,R$ implies that as $p$ increases, $L$, $C$, and $R$ change by at most 1 at every transition.
\end{observation}

\subsection{The lowest range of $p$}\label{sec:lowest-p}
Note that for even stretch $t$, when $p\in[1,\varphi)$ (where $\varphi=\frac{1+\sqrt{5}}{2}$ is the golden ratio), the above solution is simply $C=0$, $L=R=t/2$ (also a feasible dual when $p=\varphi$). For odd $t$, in the range  $p\in[1,2]$, the solution given above is slightly more complicated, but it is easy to check that in this case conditions~\eqref{eq:LCR-cond1}-\eqref{eq:LCR-cond5} also hold for the simpler solution $C=1$, $L=R=\lfloor t/2\rfloor$.

Note that for both odd and even $t$, the $(L,C,R)$-minimal spanner corresponding to these solutions is symmetric. In particular, its outer layers have the same cardinality $n_0=n_t$. Recall that this spanner spans a graph with $\ell_p$ norm $n_0^{1/p}n_t$, and so this graph is a tight example for every possible value of $\lambda$ up to the maximum possible value of $1+1/p$ (corresponding to the case where $n_0=n_t=n$). This is not the case when the spanner is not symmetric, since by gradually increasing $n_0$, eventually $n_t$ hits the upper bound of $n$ before $\lambda$ reaches its maximum possible value. As we shall see, for larger $p$ ($p>2$ for odd $t$ and $p>\varphi$ for even $t$), we will have slightly different extremal graphs when $\lambda$ is above this threshold. However, first we will consider the case of high $p$.

\subsection{General Solution for High $p$, Low $\Lambda$}

Let us now consider $(L,C,R)$-minimal spanners for $L=0$. Note that in a $(0,C,R)$ minimal spanner (assuming $n_t<n$) the only tight constraints are~\eqref{LP2:spanning},~\eqref{LP2:left-deg-norm} for all $i\in[C+R]$,~\eqref{LP2:right-deg-norm} for all $i\in[C]$,~\eqref{LP2:right-deg-positive} for $C+1\leq i\leq C+R$,~\eqref{LP2:Lambda}, and~\eqref{LP2:final-layer}. Thus, by complementary slackness, the only (possibly) non-zero dual variables in an optimal dual solution are $x,a_1,\ldots,a_{C+R},b_1,\ldots,b_C,D_{C+1},\ldots,D_{C+R},y$ and $w$. Also, all the primal variables are greater than $1$, so by complementary slackness, in an optimal dual solution, all the constraints must hold with equality. To summarize, a $(0,C,R)$ minimal spanner is optimal iff there exists a non-negative solution to the following system of linear equations:

\begin{align}
&\sum_{i=1}^{C+R}a_i+\sum_{i=1}^{C}b_i= 1&\label{LP:general-scaling2}\\
&y= x+w&\\
&p^{-1}y= p^{-1}a_1+b_1&\\
&\frac{p-1}{p}\cdot b_i= p^{-1}a_{i+1}+b_{i+1}&\forall i\in\{1,\ldots,C-1\}\\
&\frac{p-1}{p}\cdot b_{C}+D_{C+1}= p^{-1}a_{C+1}\\
& D_{i+1}= p^{-1}a_{i+1}+D_i&\forall i\in\{C+1,\ldots,C+R-1\}\\
&w= D_{C+R}\\
&x= a_i+b_i&\forall i\in\{1,\ldots,C\}\\
&x+D_i= a_i&\forall i\in\{C+1,\ldots,C+R\}
\end{align}

One can check that this system of linear equations has a unique (not necessarily non-negative) solution. For some $\eps>0$, this solution is
\begin{align*}
x&=\left((p-1)+\left(\frac{p}{p-1}\right)^{R-1}\right)\eps\\
a_i&=\left(i+p-1+\frac{i-C-1}{p-1}\left(\frac{p}{p-1}\right)^{R-1}\right)\cdot\eps&\forall i\in[C]\\
a_i&=(C+p)\left(\frac{p}{p-1}\right)^{i-C-1}\cdot\eps&\forall i\in\{C+1,\ldots,C+R\}\\
b_i&=\left(\frac{C+p-i}{p-1}\cdot\left(\frac{p}{p-1}\right)^{R-1}-i\right)\cdot\eps&\forall i\in[C]\\
D_i&=(C+p)\left(\frac{p}{p-1}\right)^{i-C-1}\cdot\eps-x&\forall i\in\{C+1,\ldots,C+R\}\\
y&=a_{C+R}=(C+p)\left(\frac{p}{p-1}\right)^{R-1}\cdot\eps\\
w&=D_{C+R}=(C+p)\left(\frac{p}{p-1}\right)^{R-1}\cdot\eps-x
\end{align*}
If all these values are non-negative (assuming $\eps>0$), then we can get a dual solution satisfying complementary slackness by choosing $\eps$ which satisfies constraint~\eqref{LP:general-scaling2}:
$$\eps=(p-1)\left((C+p)\left(\frac{p}{p-1}\right)^R-p\right)$$

Consider the conditions for non-negativity in the above solution. $x$ is non-negative by definition. The sequence $a_{1},\ldots,a_{C}$ is monotonically increasing. Thus for these values it suffices to check that $a_{1}=p-\frac{C}{p-1}\cdot\left(\frac{p-1}{p}\right)^{R-1}\geq 0$, 
or equivalently
\begin{equation}
    C\cdot\left(\frac{p}{p-1}\right)^R\leq p^2
    \label{eq:CR-cond2}
\end{equation}

The values $a_{C+1},\ldots,a_{C+R}$ are again non-negative by definition. The sequence $b_1,\ldots,b_{L+C}$ is monotonically decreasing, so for these values it suffices to check that $b_{C}=\left(\frac{p}{p-1}\right)^{R}-C\geq 0$, 
or equivalently
\begin{equation}
    \left(\frac{p}{p-1}\right)^{R}\geq C
    \label{eq:CR-cond3}
\end{equation}

Note that $y$ is non-negative by definition, and $w=D_{C+R}$. This leaves the sequence of values $D_{C+1},\ldots,D_{C+R}$, which is monotonically increasing, so it remains to check when $D_{C+1}=(C+p)-\left((p-1)+\left(\frac{p}{p-1}\right)^{R-1}\right)\geq 0$. This occurs when
\begin{equation}
    \left(\frac{p}{p-1}\right)^{R-1}\leq C+1
    \label{eq:CR-cond4}
\end{equation}

Note that conditions~\eqref{eq:CR-cond2}, \eqref{eq:CR-cond3}, and \eqref{eq:CR-cond4} are in fact the same as conditions~\eqref{eq:LCR-cond2}, \eqref{eq:LCR-cond3}, and \eqref{eq:LCR-cond4}, respectively, when we set $L=0$.

When these conditions hold simultaneously for \emph{some} values of $C,R$, it is simple to see what $C$ and $R$ must be. Since $R=t-C$, conditions~\eqref{eq:CR-cond3} and~\eqref{eq:CR-cond4} can be rewritten as
$$C\left(\frac{p}{p-1}\right)^C\leq \left(\frac{p}{p-1}\right)^t\leq (C+1)\left(\frac{p}{p-1}\right)^{C+1}$$
Thus, the value $\left(\frac{p}{p-1}\right)^t$ occurs in exactly one of the disjoint intervals $\left\{I_C\mid C\in\nats\right\}$, where $I_C=\left[C\left(\frac{p}{p-1}\right)^C,(C+1)\left(\frac{p}{p-1}\right)^{C+1}\right)$, and we we choose the corresponding value of $C$ (which also determines $R=t-C$).

It remains to show that there exist such solutions, and that the range of $p$ for which there exists such a solution, together with the range for which there exists a solution as in the previous section, cover all possible $p\in[1,\infty)$.

To see that there exists a solution for \emph{some} value of $p$, note that in the solution in the previous section, we had $L\approx(t-p-p\ln p)/2$. Thus, for sufficiently large $p$, $L$ can no longer be positive. However, by Observation~\ref{obs:LCR-continuity}, $L$ cannot jump from being strictly positive to being strictly negative. For the minimum value of $p$ such that in the previous section $L$ can no longer be positive, we have $L=0$. In particular, this means that for this value of $p$ conditions~\eqref{eq:LCR-cond2}, \eqref{eq:LCR-cond3}, and \eqref{eq:LCR-cond4} hold for $L=0$, meaning, conditions~\eqref{eq:CR-cond2}, \eqref{eq:CR-cond3}, and \eqref{eq:CR-cond4} hold.

Let us see that in fact a solution (with $L=0$) exists for all values of $p$ greater than or equal to the above value. This can be seen via a simple monotonicity argument. We claim that if such a solution exists for some $p$, with $C\leq p$, then such a solution exists for all $p'\geq p$ as well (note that $C\leq p$ by our choice of $C$ for the initial value of $p$). First, note that as long as $C$ does not change, we can increase $p$ and the inequalities in~\eqref{eq:CR-cond2} and~\eqref{eq:CR-cond4} will be strengthened. However, at some point~\eqref{eq:CR-cond3} may be violated. Suppose we reach $p$ such that condition~\eqref{eq:CR-cond3} becomes tight. That is, we have $(p/(p-1))^{t-C}=C$. Then it is easy to see that conditions~\eqref{eq:CR-cond2}-\eqref{eq:CR-cond4} hold for $C'=C-1$. Indeed, if~\eqref{eq:CR-cond2} holds for $p,C$, then \begin{align*}
    C'\left(\frac{p}{p-1}\right)^{t-C'}=(C-1)\cdot\frac{p}{p-1}\left(\frac{p}{p-1}\right)^{t-C}&\leq \frac{C-1}{C}&\text{by~\eqref{eq:CR-cond2}}\\\cdot\frac{p}{p-1}\cdot p^2&\leq p^2&\text{since $C\leq p$},
\end{align*}
giving us condition~\eqref{eq:CR-cond2} for $p,C'$.
Since condition~\eqref{eq:CR-cond3} holds (in fact with equality) for $p,C$, then it trivially holds for $p,C'$ by monotonicity. Finally, since~\eqref{eq:CR-cond3} holds with equality, we have
$$\left(\frac{p}{p-1}\right)^{t-C'-1}=\left(\frac{p}{p-1}\right)^{t-C}=C=C'+1,$$ which gives condition~\eqref{eq:CR-cond4} for $p,C'$.

Note that since $t\geq 2$ and $p/(p-1)>1$, condition~\eqref{eq:CR-cond4} cannot be satisfied for $C=0$. Thus, for the very highest range of $p$ we have $C=1$, $R=t-1$.

\subsection{Handling Large Values of $\Lambda$}
As we've seen, for the lowest range of $p$ (specifically, $p\leq 2$ for odd $t$, and $p\leq\varphi$ for even $t$), for every $t$ there is a single setting of $(L,C,R)$ such that the tight lower bound is given by an $(L,C,R)$-minimal spanner, for every possible value of $\Delta$.

However, for larger values of $p$, this is no longer the case. Note that for larger values of $p$, our solution always has $C\geq 1$ and $R>L$. A careful examination of the properties of an $(L,C,R)$-minimal spanner shows that if its central layers have size $n_L(=n_{L+1}=\ldots=n_{L+C})$, then its outer layers have sizes $n_0=n_C^{1+\frac{p}{c}\left(1-\left(\frac{p-1}{p}\right)^L\right)}$ and $n_t=n_{L+C+R}=n_C^{1+\frac{p}{c}\left(1-\left(\frac{p-1}{p}\right)^R\right)}$. In the range where $L<R$, this means that $n_0<n_t$. In particular, we cannot increase the size of such an $(L,C,R)$-minimal spanner past the point where $n_t=n$. A simple calculation shows that this occurs when $\lambda=1+\frac{1}{p}\left(1+\frac{p}{c}\left(1-\left(\frac{p-1}{p}\right)^L\right)\right)\left(1+\frac{p}{c}\left(1-\left(\frac{p-1}{p}\right)^R\right)\right)^{-1}<1+\frac{1}{p}$. We call this the \emph{nice} range of $\lambda$ (for the corresponding choice of $p,t$). Beyond this value of $\lambda$, our tight lower bound examples are no longer $(L,C,R)$-minimal spanners, but some slight variant.

\begin{definition}
A \emph{skewed} $(L,C,R)$ minimal spanner is a layered graph with $L+C+R+1$ layers of size $n_0\geq n_1\geq\ldots\geq n_L=n_{L+1}=\ldots=n_{L+C}\leq n_{L+C+1}\leq\ldots\leq n_{L+C+R}$, similar to an $(L,C,R)$ minimal spanner, with the following exceptions:
\begin{itemize}
    \item In a \emph{left-skewed} $(L,C,R)$ minimal spanner, we have $d_{L}\geq 1$ (not necessarily equal 1), while in a \emph{right-skewed} $(L,C,R)$ minimal spanner, we have $n_{L+C}d_{L+C+1}\geq n_{L+C+1}$ (that is, the degree from layer $(L+C+1)$ back into layer $(L+C)$ may be greater than 1). We do not allow strict inequality in both simultaneously.
    \item We no longer necessarily have $d_{L+1}=\ldots=d_{L+C}=n_L^{1/C}$. Rather, if we define the \emph{skew degree} $\tilde{d}$ to be $d_L$ in a left-skewed spanner, and $n_{L+C}d_{L+C+1}/n_{L+C+1}$ in a right-skewed spanner, then these degrees are $d_{L+1}=\ldots=d_{L+C}=(n_L/\tilde{d})^{1/C}$. That is, they are chosen such that $\tilde d\cdot\prod_{i={L+1}}^{L+C}d_i=n_L$.
\end{itemize}
\end{definition}

The idea of a skewed $(L,C,R)$ minimal spanner is that it allows us to move smoothly between different $(L,C,R)$ minimal spanners by gradually changing the subgraph between two consecutive layers (and possibly rescaling the entire graph). In particular, note that an $(L,C,R)$ minimal spanner could be considered a left-skewed $(t-C'-R,C',R)$ minimal spanner as well as a right-skewed $(L,C',t-L-C')$ minimal spanner for either $C'\in\{C-1,C\}$.

\subsubsection{Conditions for Optimality of a Left-Skewed Spanner, for Low $p$}
Here we examine for which parameter settings a left-skewed $(\tilde L,C,\tilde R)$ minimal spanner with $n_t=n$ is an optimal solution, for $\tilde L,C,\tilde R>0$. Note that in such a spanner, the only tight constraints are~\eqref{LP2:spanning},~\eqref{LP2:left-deg-norm} for $\tilde L+1\leq i\leq \tilde L+C+\tilde R$,~\eqref{LP2:right-deg-norm} for $i\in[\tilde L+C]$,~\eqref{LP2:right-deg-positive} for $\tilde L+C+1\leq i\leq \tilde L+C+\tilde R$,~\eqref{LP2:Lambda},~\eqref{LP2:final-layer}, and~\eqref{LP2:n}. Thus, by complementary slackness, the only (possibly) non-zero dual variables in an optimal dual solution are $x,a_{\tilde L+1},\ldots,a_{\tilde L+C+\tilde R},b_1,\ldots,b_{\tilde L+C},D_{\tilde L+C+1},\ldots,D_{\tilde L+C+\tilde R},y$, $w$, and $s$. Also, all the primal variables are greater than $1$ except for $d_1,\ldots,d_{\tilde L-1}$, so by complementary slackness, in an optimal dual solution, all the constraints corresponding to other primal variables must hold with equality.\footnote{Note that constraint~\eqref{LP2:left-deg-norm} could be tight for $i=\tilde L$, or we could have $d_L=1$, but this does not hurt our argument. It would only mean that the conditions for complementary slackness we present will be sufficient, not necessary.} To summarize, a left-skewed $(\tilde L,C,\tilde R)$ minimal spanner is optimal iff there exists a non-negative solution to the following system of equations and inequalities:

\begin{align}
&\sum_{i=\tilde L+1}^{\tilde L+C+\tilde R}a_i+\sum_{i=1}^{\tilde L+C}b_i= 1&\label{LP:general-scaling3}\\
&y= x+w&\\
&p^{-1}y= b_1&\\
&\frac{p-1}{p}\cdot b_i= b_{i+1}&\forall i\in[\tilde L-1]\\
&\frac{p-1}{p}\cdot b_i= p^{-1}a_{i+1}+b_{i+1}&\forall i\in\{\tilde L,\ldots, \tilde L+C-1\}\\
&\frac{p-1}{p}\cdot b_{\tilde L+C}+D_{\tilde L+C+1}= p^{-1}a_{\tilde L+C+1}\\
& D_{i+1}= p^{-1}a_{i+1}+D_i&\forall i\in\{\tilde L+C+1,\ldots,\tilde L+C+\tilde R-1\}\\
&w= D_{\tilde L+C+\tilde R}+s\\
&x\leq b_i&\forall i\in[\tilde L-1]\label{LP:general-ineq3}\\
&x=b_{\tilde L}\\
&x= a_i+b_i&\forall i\in\{\tilde L+1,\ldots,\tilde L+C\}\\
&x+D_i= a_i&\forall i\in\{\tilde L+C+1,\ldots,\tilde L+C+\tilde R\}
\end{align}
It is not hard to check that there is a unique (not necessarily non-negative) to the above constraints excluding~\eqref{LP:general-ineq3}. For some value of $x>0$ (that can be scaled so as to satisfy constraint~\eqref{LP:general-scaling3}), this solution is
\begin{align*}
    a_i&=\frac{i-\tilde L}{p-1}\cdot x&\forall i\in\{\tilde L+1,\ldots,\tilde L+C\}\\
    a_i&=\frac{C+1}{p-1}\cdot\left(\frac{p}{p-1}\right)^{i-(\tilde L+C+1)}x&\forall i\in\{\tilde L+C+1,\ldots,\tilde L+C+\tilde R\}\\
    b_i&=\left(\frac{p}{p-1}\right)^{\tilde L-i}x&\forall i\in[\tilde L]\\
    b_i&=\left(1-\frac{i-\tilde L}{p-1}\right)x&\forall i\in\{\tilde L,\ldots,\tilde L+C\}\\
    D_i&=\left(\frac{C+1}{p-1}\cdot\left(\frac{p}{p-1}\right)^{i-(\tilde L+C+1)}-1\right)x&\forall i\in\{\tilde L+C+1,\ldots,\tilde L+C+\tilde R\}\\
    y&=p\left(\frac{p}{p-1}\right)^{\tilde L-1}x\\
    w&=\left(p\left(\frac{p}{p-1}\right)^{\tilde L-1}-1\right)x\\
    s&=\left(p\left(\frac{p}{p-1}\right)^{\tilde L-1}-\frac{C+1}{p-1}\cdot\left(\frac{p}{p-1}\right)^{\tilde R-1}\right)x
\end{align*}
The non-negativity of most of the above variables follows by definition, as does constraint~\eqref{LP:general-ineq3}. The non-negativity of the $D_i$s is equivalent to the non-negativity of $D_{\tilde L+C+1}$, which follows iff
\begin{equation}
    C\geq p-2
\end{equation}
The only additional variable which is not trivially non-negative is $s$, which is non-negative iff
\begin{equation}
    (C+1)\left(\frac{p}{p-1}\right)^{\tilde R-\tilde L-1}\leq p^2
    \label{eq:left-skew}
\end{equation}

\subsubsection{Conditions for Optimality of a Right-Skewed Spanner, for Low $p$}
Here we examine for which parameter settings a right-skewed $(L,C,R)$ minimal spanner with $n_t=n$ is an optimal solution, for $L,C,R>0$. Note that in such a spanner, the only tight constraints are~\eqref{LP2:spanning},~\eqref{LP2:left-deg-norm} for $L+1\leq i\leq L+C+R$,~\eqref{LP2:right-deg-norm} for $i\in[L+C]$,~\eqref{LP2:right-deg-positive} for $L+C+2\leq i\leq L+C+R$,~\eqref{LP2:Lambda},~\eqref{LP2:final-layer}, and~\eqref{LP2:n}. Thus, by complementary slackness, the only (possibly) non-zero dual variables in an optimal dual solution are $x,a_{L+1},\ldots,a_{L+C+R},b_1,\ldots,b_{L+C},D_{L+C+2},\ldots,D_{L+C+R},y$, $w$, and $s$. Also, all the primal variables are greater than $1$ except for $d_1,\ldots,d_{L}$, so by complementary slackness, in an optimal dual solution, all the constraints corresponding to other primal variables must hold with equality.\footnote{Note that constraint~\eqref{LP2:right-deg-positive} or constraint~\eqref{LP2:right-deg-norm} could be tight for $i=L+C+1$, but as before, this does not hurt our argument.} To summarize, a right-skewed $(L,C,R)$ minimal spanner is optimal iff there exists a non-negative solution to the following system of equations and inequalities:

\begin{align}
&\sum_{i=L+1}^{L+C+R}a_i+\sum_{i=1}^{L+C}b_i= 1&\label{LP:general-scaling4}\\
&y= x+w&\\
&p^{-1}y= b_1&\\
&\frac{p-1}{p}\cdot b_i= b_{i+1}&\forall i\in[L-1]\\
&\frac{p-1}{p}\cdot b_i= p^{-1}a_{i+1}+b_{i+1}&\forall i\in\{L,\ldots, L+C-1\}\\
&\frac{p-1}{p}\cdot b_{L+C}= p^{-1}a_{L+C+1}\\
& D_{L+C+2}=p^{-1}a_{L+C+2}\\
& D_{i+1}= p^{-1}a_{i+1}+D_i&\forall i\in\{L+C+2,\ldots,L+C+R-1\}\\
&w= D_{L+C+R}+s\\
&x\leq b_i&\forall i\in[L]\label{LP:general-ineq4}\\
&x= a_i+b_i&\forall i\in\{L+1,\ldots,L+C\}\\
&x=a_{L+C+1}\\
&x+D_i= a_i&\forall i\in\{L+C+2,\ldots,L+C+R\}
\end{align}

It is not hard to check that there is a unique (not necessarily non-negative) to the above constraints excluding~\eqref{LP:general-ineq4}. For some value of $x>0$ (that can be scaled so as to satisfy constraint~\eqref{LP:general-scaling4}), this solution is
\begin{align*}
    a_i&=\left(1-\frac{L+C+1-i}{p-1}\right)\cdot x&\forall i\in\{L+1,\ldots,L+C\}\\
    a_i&=\left(\frac{p}{p-1}\right)^{i-(L+C+1)}x&\forall i\in\{L+C+1,\ldots,L+C+R\}\\
    b_i&=\frac{C+1}{p-1}\cdot\left(\frac{p}{p-1}\right)^{L-i}x&\forall i\in[L]\\
    b_i&=\frac{L+C+1-i}{p-1}\cdot x&\forall i\in\{L,\ldots,L+C\}\\
    D_i&=\left(\left(\frac{p}{p-1}\right)^{i-(L+C+1)}-1\right)x&\forall i\in\{L+C+2,\ldots,L+C+R\}\\
    y&=(C+1)\cdot\left(\frac{p}{p-1}\right)^{L}x\\
    w&=\left((C+1)\left(\frac{p}{p-1}\right)^{L}-1\right)x\\
    s&=\left((C+1)\left(\frac{p}{p-1}\right)^{L}-\left(\frac{p}{p-1}\right)^{R-1}\right)x
\end{align*}
The non-negativity of most of the above variables follows by definition. By monotonicity, constraint~\eqref{LP:general-ineq4} follows for all $i\in[L]$ iff it follows for $i=L$, which is when
\begin{equation}
    C\geq p-2
\end{equation}
The only variable which is not trivially non-negative is $s$, which is non-negative iff
\begin{equation}
    \left(\frac{p}{p-1}\right)^{R-L-1}\leq C+1
    \label{eq:right-skew}
\end{equation}

\subsubsection{Conditions for Optimality of a Right-Skewed Spanner, for High $p$}
Here we examine for which parameter settings a right-skewed $(0,C,R)$ minimal spanner with $n_t=n$ is an optimal solution. Note that in such a spanner, the only tight constraints are~\eqref{LP2:spanning},~\eqref{LP2:left-deg-norm} for $i\leq [C+R]$,~\eqref{LP2:right-deg-norm} for $i\in[C]$,~\eqref{LP2:right-deg-positive} for $C+2\leq i\leq C+R$,~\eqref{LP2:Lambda},~\eqref{LP2:final-layer}, and~\eqref{LP2:n}. Thus, by complementary slackness, the only (possibly) non-zero dual variables in an optimal dual solution are $x,a_{1},\ldots,a_{C+R},b_1,\ldots,b_{C},D_{C+2},\ldots,D_{C+R},y$, $w$, and $s$. Also, all the primal variables are greater than $1$, so by complementary slackness, in an optimal dual solution, all the constraints must hold with equality.\footnote{As before, constraint~\eqref{LP2:right-deg-positive} or constraint~\eqref{LP2:right-deg-norm} could be tight for $i=C+1$.} To summarize, a right-skewed $(0,C,R)$ minimal spanner is optimal iff there exists a non-negative solution to the following system of linear equations:

\begin{align}
&\sum_{i=1}^{C+R}a_i+\sum_{i=1}^{C}b_i= 1&\label{LP:general-scaling5}\\
&y= x+w&\\
&p^{-1}y= p^{-1}a_1+b_1&\\
&\frac{p-1}{p}\cdot b_i= p^{-1}a_{i+1}+b_{i+1}&\forall i\in[C-1]\\
&\frac{p-1}{p}\cdot b_{C}= p^{-1}a_{C+1}\\
& D_{C+2}=p^{-1}a_{C+2}\\
& D_{i+1}= p^{-1}a_{i+1}+D_i&\forall i\in\{C+2,\ldots,C+R-1\}\\
&w= D_{C+R}+s\\
&x= a_i+b_i&\forall i\in[C]\\
&x=a_{C+1}\\
&x+D_i= a_i&\forall i\in\{C+2,\ldots,C+R\}
\end{align}

It is not hard to check that there is a unique (not necessarily non-negative) to the above system of equations. For some value of $x>0$ (that can be scaled so as to satisfy constraint~\eqref{LP:general-scaling5}), this solution is
\begin{align*}
    a_i&=\left(1-\frac{C+1-i}{p-1}\right)\cdot x&\forall i\in[L+C]\\
    a_i&=\left(\frac{p}{p-1}\right)^{i-(C+1)}x&\forall i\in\{C+1,\ldots,C+R\}\\
    b_i&=\frac{C+1-i}{p-1}\cdot x&\forall i\in[C]\\
    D_i&=\left(\left(\frac{p}{p-1}\right)^{i-(C+1)}-1\right)x&\forall i\in\{C+2,\ldots,C+R\}\\
    y&=(C+1)x\\
    w&=Cx\\
    s&=\left(C+1-\left(\frac{p}{p-1}\right)^{R-1}\right)x
\end{align*}
The non-negativity of most of the above variables follows by definition. 
The only variable which is not trivially non-negative is $s$, which is non-negative iff
\begin{equation}
    \left(\frac{p}{p-1}\right)^{R-1}\leq C+1
    \label{eq:high-p-skew}
\end{equation}

\subsection{Optimal Solutions for High $\Lambda$}
We now describe optimal solutions in the ``not nice" region of $\lambda$. That is, in the region where $\lambda$ is at least at the threshold where an optimal $(L,C,R)$-minimal spanner as described earlier no longer exists, due to the size of the final layer, $n_t$. Our examples will be based on simple manipulations of the optimal $(L,C,R)$ minimal spanner for the corresponding values of $p,t$.

First, consider the case of low $p$ (when there exists an optimal solution with $L>0$). Recall that a $(L,C,R)$ minimal spanner is also a right-skewed $(L,C,R)$ minimal spanner. Not that the conditions for optimality for a right-skewed spanner are already implied by conditions~\eqref{eq:LCR-cond4} and~\eqref{eq:LCR-cond5}. Thus, we can interpolate between an $(L,C,R)$ minimal spanner and an $(L,C+1,R-1)$ minimal spanner (which is also a right-skewed $(L,C,R)$ minimal spanner), by setting $n_t=n$ and $n_{L+C+1}=n_{L+C}^\alpha$ for all possible $\alpha\in[1,1+1/C]$, and all of these will be optimal solutions.

Now, if $L=R-1$ then the final graph in this interpolation (the $(L,C+1,R-1)$-minimal spanner) has $n_0=n_t$ and thus is an optimal solution for $\lambda=1+1/p$, and thus the intermediate graphs cover all the remaining possible values of $\lambda$. Otherwise, we note that this graph is also a left-skewed $(\tilde L,C,\tilde R)$-minimal spanner for $\tilde L=L+1$, $\tilde R=R-1$. To see that such any such spanner with these parameters will be optimal here, note again that $C\geq p-2$, and that~\eqref{eq:left-skew} follows from~\eqref{eq:LCR-cond4}, since
\begin{align*}
    (C+1)\left(\frac{p}{p-1}\right)^{\tilde R-\tilde L-1}&=(C+1)\left(\frac{p}{p-1}\right)^{R-L-3}\\
    &=(C+1)\left(\frac{p-1}{p}\right)^2\cdot\left(\frac{p}{p-1}\right)^{R-L-1}\\
    &\leq (C+1)^2\left(\frac{p-1}{p}\right)^2&\text{by~\eqref{eq:LCR-cond4}}\\
    &\leq (p+1)^2\left(\frac{p-1}{p}\right)^2&\text{by~\eqref{eq:LCR-cond5}}\\
    &=\left(\frac{p^2-1}{p}\right)^2< p^2
\end{align*}
Thus, we can now interpolate between an $(L,C+1,R-1)$ minimal spanner and an $(L+1,C,R-1)$ minimal spanner (both of which are right-skewed $(L+1,C,R-1)$ minimal spanner), by setting $n_t=n$ and $n_{L-1}=n_{L}^\alpha$ for all possible $\alpha\in[1,1+1/C]$, and all of these will be optimal solutions. Once again, if $L+1=R-1$, then we have covered the entire range of $\lambda$, as before. Otherwise, we can continue to repeatedly alternate between right-skewed and left-skewed spanners as above until $n_0=n_1$. Note that all of these graphs will be optimal by the same argument, as the value of $C$ never changes, and condition~\eqref{eq:LCR-cond4} will continue to hold, as we only decrease the value of $R-L$.

Finally, we note that a similar (though much simpler) argument holds for the case of high $p$ (when $L=0$). Since condition~\eqref{eq:high-p-skew} is exactly the same as condition~\eqref{eq:CR-cond4}, the conditions for optimality of a skewed $(0,C,R)$ minimal spanner are already satisfied for an $(L,C,R)$ minimal spanner when $n_t=n$. Furthermore, this condition also holds by monotonicity for higher values of $C$ (and lower values of $R=t-C$). Thus, as before, for every $C\leq C'\leq t-1$ we can interpolate between a $(0,C',t-C')$ minimal spanner and a $(0,C'+1,t-C'-1)$ minimal spanner using right-skewed $(0,C',t-C')$ minimal spanners, all of which will be optimal, and this will cover the entire range of $\lambda$.

\section{Future Work}
In this paper we have initiated the study of graph spanners with cost defined by the $\ell_p$-norm of the degree vector, since this provides an interesting interpolation between the $\ell_1$-norm (only caring about the number of edges) and the $\ell_{\infty}$-norm (only caring about the maximum degree).  But we have only scratched the surface: many of the hundreds of results on graph spanners can be extended or reexamined with respect to the $\ell_p$-norm.  There are also some very interesting direct extensions of this paper that would be interesting to study.  In particular, we showed that the approximation ratio achieved by the greedy algorithm is strictly better than the generic guarantee for the $\ell_2$-norm with stretch $3$, unlike the $\ell_1$ and $\ell_{\infty}$ norms.  This suggests further study of the greedy algorithm in general, but also suggests extending the recent line of work on approximation algorithms for graph spanners (mostly using convex relaxations and rounding) to general $\ell_p$-norms.  The approaches taken for the $\ell_1$-norm in the past~\cite{DK11-stoc,DK11-podc,BBMRY13,DZ16} have been quite different from the approaches used for the $\ell_{\infty}$-norm~\cite{KP98,CDK12,CD14}; is there a way of interpolating between them to get even better approximations for intermediate $\ell_p$-norms?

\bibliographystyle{plainurl}
\bibliography{refs}

\appendix

\section{Tightness of Upper Bound} \label{app:UB-tight}

We show that, assuming the Erd\H{o}s girth conjecture, our upper bound (Theorem~\ref{thm:upper-main}) is tight even when parameterizing by $\Lambda$ in addition to $n$.  More formally, we prove the following theorem:
\begin{theorem} \label{thm:UB-tight}
Assuming the Erd\H{o}s girth conjecture, $\mathrm{UB}_{2k-1}^p(n, \Lambda) \geq \Omega(\min(\max(n, n^{\frac{k+p}{kp}}), \Lambda))$ for all $k \geq 2$, $p \geq 1$, and $\Omega(n^{1/p}) \leq \Lambda \leq O(n^{\frac{1+p}{p}})$.
\end{theorem}
As a simple corollary, if we do not parameterize by $\Lambda$ we get the following straightforward complement to Theorem~\ref{thm:upper-main}:
\begin{corollary}
Assuming the Erd\H{o}s girth conjecture, for every $k \geq 2$ and $p \geq 1$ there is a connected graph $G$ such that every $(2k-1)$-spanner $H$ of $G$ has $\|H\|_p \geq \Omega(\max(n, n^{\frac{k+p}{kp}}))$. 
\end{corollary}

We now prove Theorem~\ref{thm:UB-tight}.  We break into two cases depending on $p$, and then for each case break into two more cases depending on $\Lambda$.  First, suppose that $p \geq k/(k-1)$, so $n \geq n^{\frac{k+p}{kp}}$.  If $\Lambda \leq n$, let $G$ be a graph consisting of a star with $\Lambda$ leaves together with a path of length $n - \Lambda - 1$, where one endpoint of the path is also adjacent to an arbitrary leaf of the star.  Then $G$ clearly has $n$ nodes and $\|G\|_p = \Theta(\Lambda)$.  Moreover, since $G$ is a tree, the only $(2k-1)$-spanner of $G$ is $G$ itself.  Thus in this case $\mathrm{UB}_{2k-1}^p(n, \Lambda) \geq \Omega(\Lambda) = \Omega(\min(\Lambda, n))$.  On the other hand, if $\Lambda > n$, then let $G$ be a clique on $\Lambda^{\frac{p}{1+p}}$ nodes, together with a star with $n$ leaves (with an arbitrary vertex of the clique adjacent to an arbitrary vertex of the star to make $G$ connected).  Then it is easy to see that $G$ has $\Theta(n)$ nodes and $\|G\|_p = \Theta(\Lambda)$, and moreover that any $(2k-1)$-spanner of the tree must include every edge of the star.  Thus we get that $\mathrm{UB}_{2k-1}^p(n, \Lambda) \geq \Omega(n) = \Omega(\min(\Lambda, n))$ in this case.  

Now suppose that $1 \leq p \leq k/(k-1)$, so $n^{\frac{k+p}{kp}} \geq n$.  Let $H_n$ be a graph from the Erd\H{o}s girth conjecture: a graph with $\Theta(n)$ nodes that is regular with degree $\Theta(n^{1/k})$ and has girth at least $2k+1$ (note that such graphs are known to exist for particular values of $k$ such as $k=2,3,5$~\cite{Wenger91}). Note that $\|H_n\|_p = \Theta((n \cdot (n^{1/k})^p)^{1/p}) = \Theta(n^{\frac{k+p}{kp}})$.  If $\Lambda \leq n^{\frac{k+p}{kp}}$, then let $G$ be an arbitrary subgraph of $H_n$ with $\|G\|_p = \Lambda$.  Since $G$ has girth at least $2k+1$, the only $(2k-1)$-spanner of $G$ is $G$ itself.  Thus in this case $\mathrm{UB}_{2k-1}^p(n, \Lambda) \geq \Omega(\Lambda) = \Omega(\min(\Lambda, n^{\frac{k+p}{kp}}))$.  On the other hand, suppose that $\Lambda > n^{\frac{k+p}{kp}}$.  Then we can build $G$ by building a clique of size $\Lambda^{\frac{1+p}{p}}$ and combining this with $H_{n/2}$, with one arbitrary edge between the clique and $H_{n/2}$.  Then $G$ has $\Theta(n)$ nodes and $\|G\|_p = \Theta(\Lambda)$, and any $(2k-1)$-spanner of $G$ must include every edge of $H_{n/2}$.  Thus $\mathrm{UB}_{2k-1}^p(n, \Lambda) \geq \Omega(n^{\frac{k+p}{kp}}) = \Omega(\min(\Lambda, n^{\frac{k+p}{kp}}))$.

\end{document}